\documentclass[12pt,reqno]{amsart}
\usepackage{graphicx}
\usepackage{amssymb,amsmath}
\usepackage{amsthm}
\usepackage{color}
\usepackage[pdf]{pstricks}
\usepackage{hyperref}

\usepackage{amssymb}
\usepackage{epsfig}
\usepackage{extarrows}
\usepackage{amsfonts}
\usepackage{graphicx}

\newcommand{\R}{\mathbb{R}}

\newtheorem{remark}{Remark}
\newtheorem{proposition}{Proposition}
\newtheorem{theorem}{Theorem}

\begin{document}
	
\title[KdV breathers]{\bf KdV breathers on a cnoidal wave
  background}
	
\author{Mark Hoefer}
\address[M. Hoefer]{Department of Applied Mathematics,  University of Colorado at Bolder, USA}
\email{hoefer@colorado.edu}
	
\author{Ana Mucalica}
\address[A. Mucalica]{Department of Mathematics and Statistics, McMaster University,	Hamilton, Ontario, Canada, L8S 4K1}
\email{mucalica@mcmaster.ca} 
	
\author{Dmitry E. Pelinovsky}
\address[D. Pelinovsky]{Department of Mathematics and Statistics, McMaster University, Hamilton, Ontario, Canada, L8S 4K1}
\email{dmpeli@math.mcmaster.ca}

\begin{abstract}
  Using the Darboux transformation for the Korteweg-de Vries equation,
  we construct and analyze exact solutions describing the interaction
  of a solitary wave and a traveling cnoidal wave.  Due to their
  unsteady, wavepacket-like character, these wave patterns are
  referred to as breathers.  Both elevation (bright) and depression
  (dark) breather solutions are obtained.  The nonlinear dispersion
  relations demonstrate that the bright (dark) breathers propagate
  faster (slower) than the background cnoidal wave. Two-soliton
  solutions are obtained in the limit of degeneration of the cnoidal
  wave. In the small amplitude regime, the dark breathers are
  accurately approximated by dark soliton solutions of the nonlinear
  Schr\"odinger equation.  These results provide insight into recent
  experiments on soliton-dispersive shock wave interactions and
  soliton gases.
\end{abstract}

\date{\today}
\maketitle

\section{Introduction}

The localized and periodic traveling wave solutions of the Korteweg-de
Vries (KdV) equation are so ubiquitous and fundamental to nonlinear
science that their names, ``soliton'' and ``cnoidal wave,'' have
achieved a much broader usage, representing localized and periodically
extended traveling wave solutions across a wide range of nonlinear
evolutionary equations.  Consequently, it is natural and important to
consider their interactions.  While the traditional notion of linear
superposition cannot be used, the complete integrability of the KdV
equation implies a nonlinear superposition principle.  For example,
soliton interactions can be described by exact $N$-soliton solutions,
which can be constructed by successive Darboux transformations
\cite{Matveev}.  By utilizing solutions of the spectral problem for
the stationary Schr\"odinger equation and the temporal evolution equation whose
compatibility is equivalent to solving the KdV equation, the Darboux
transformation achieves a nonlinear superposition principle by
effectively ``adding'' one soliton to the base solution.  In the
spectral problem, the soliton appears as an additional eigenvalue that
is added to the spectrum of the base solution.  

Compared to soliton
interactions, soliton-cnoidal wave interactions have not been explored
in as much detail.  The purpose of this paper is to apply the Darboux
transformation to the cnoidal wave solution of the KdV equation in
order to obtain the nonlinear superposition of a single soliton and a
cnoidal wave.  These exact solutions, expressed in terms of Jacobi
theta functions and elliptic integrals, represent the interactions of
a soliton and a cnoidal wave.

The motivation for this study comes from recent experiments and
analysis of the interaction of solitons and dispersive shock waves
(DSWs) \cite{Maiden,Sprenger,Cole}. The DSWs can be viewed as
modulated cnoidal waves \cite{Flaschka,ElHoefer} so that soliton-DSW
interaction is analogous to soliton-cnoidal wave interaction. Two
different types of soliton-DSW interaction dynamics were observed in
\cite{Maiden}. When a soliton completely passes through a DSW, the
nature of the interaction gives rise to an elevation (bright)
nonlinear wavepacket.  When a soliton becomes embedded or trapped
within a DSW, the trapped soliton resembles a depression (dark)
nonlinear wavepacket.  Similar transmission and trapping scenarios
were analyzed for solitons interacting with rarefaction waves
\cite{AbCole,Muc}.

Breathers are localized, unsteady solutions that exhibit two distinct
time scales or velocities; one associated with propagation and the
other with internal oscillations. A canonical model equation that
admits breather solutions is the focusing modified Korteweg-de Vries
(mKdV) equation. These solutions can be interpreted as bound states of
two soliton solutions \cite{AblowitzSegur,Clarke}.  It is in a similar
spirit that we regard as a breather, the soliton-cnoidal wave
interactions considered here. Such wavepacket solutions are
propagating, nonlinear solutions with internal oscillations.

Among our main results, we find two distinct varieties of exact
solutions of the KdV equation, corresponding to elevation (bright) or
depression (dark) breathers interacting with the cnoidal wave
background. These breathers are topological because they impart a
phase shift to the cnoidal wave. We show that bright breathers
propagate faster than the cnoidal wave, whereas dark breathers move
slower. Furthermore, bright breathers of sufficiently small amplitude
exhibit a negative phase shift, whereas bright breathers of
sufficiently large amplitude exhibit a positive phase shift. On the
other hand, dark breathers with the strongest localization have a
positive phase shift.  Small amplitude dark breathers can exhibit
either a negative or positive phase shift. Each breather solution is
characterized by its position and a spectral parameter, determining a
nonlinear dispersion relation, which uniquely relates the breather
velocity to the breather phase shift.

Exact solutions representing soliton-cnoidal wave interactions have
previously been constructed using other solution methods. The first
result was developed in \cite{Kuznetsov} within the context of the
stability analysis of a cnoidal wave of the KdV equation. The authors
used the Marchenko equation of the inverse scattering transform and
obtained exact solutions for ``dislocations'' of the cnoidal
wave. More special solutions for soliton-cnoidal wave interactions
were obtained in \cite{Hu} by using the nonlocal symmetries of the KdV
equation.  These solutions are expressed in a closed form as integrals
of Jacobi elliptic functions, but they do not represent the most
general exact solutions for soliton-cnoidal wave interactions.

Quasi-periodic (finite-gap) solutions and solitons on a quasi-periodic
background have been obtained as exact solutions of the KdV equation
by using algebro-geometric methods \cite{Bobenko,Gesztesy}. In the
limit of a single gap, such solutions describe interactions of
solitons with a cnoidal wave.  By using the degeneration of
hyperelliptic curves and Sato Grassmannian theory, mixing between
solitons and quasi-periodic solutions was obtained recently in
\cite{Nakan1} based on \cite{Nakan2}, not only for the KdV equation
but also for the KP hierarchy of integrable equations.  Finally, in a
very recent preprint \cite{Bertola}, inspired by recent works on
soliton gases \cite{Girotti,CongyEl}, the degeneration of
quasi-periodic solutions was used to construct multisoliton-cnoidal
wave interaction solutions.

Compared to previous work, which primarily involve Weierstrass
functions with complex translation parameters, we give explicit
solutions in terms of Jacobi elliptic functions with real-valued
parameters. This approach allows us to clarify the nature of
soliton-cnoidal wave interactions, plot their corresponding
properties, and analyze the exact solutions in various limiting
regimes. We also demonstrate that the Darboux transformation provides
a more straightforward method for obtaining these complicated
interaction solutions compared to the degeneration methods used in \cite{Nakan1,Bertola}.

The paper is organized as follows. The main results are formulated in
Section \ref{sec-1} and illustrated graphically. In Section
\ref{sec-2}, we introduce the normalized cnoidal wave solution with
one parameter.  Symmetries of the KdV equation are then introduced
that can be used to generate the more general family of cnoidal waves
with four arbitrary parameters.  Eigenfunctions of the stationary
Schr\"{o}dinger equation with the normalized cnoidal wave potential
are reviewed in Section \ref{sec-3}. The time evolution of the
eigenfunctions is obtained in Section \ref{sec-4}. In Section
\ref{sec-5}, the Darboux transformation is used to generate breather
solutions to the KdV equation. Properties of bright and dark breathers
are explored in Sections \ref{sec-6} and \ref{sec-7},
respectively. The paper concludes with Section \ref{sec-8}.

\section{Main results}
\label{sec-1}

We take the Korteweg--de Vries (KdV) equation in the normalized form
\begin{equation}
  \label{kdv}
  u_t+6uu_x+u_{xxx}=0,
\end{equation}
where $t$ is the evolution time, $x$ is the spatial coordinate for
wave propagation, and $u$ is the fluid velocity. As is well-known
\cite{Gardner}, every smooth solution $u(x,t)$ of the KdV equation
(\ref{kdv}) is the compatibility condition of the stationary
Schr\"{o}dinger equation
\begin{equation}
\label{stat-Sch}
(-\partial_x^2 - u) v = \lambda v
\end{equation}
and the time evolution problem 
\begin{equation}
\label{time-evol}
v_t=  (4\lambda - 2u) v_x + u_x v,
\end{equation} 
where $\lambda$ is the $(x,t)$-independent spectral parameter.

The normalized traveling cnoidal wave of the KdV equation (\ref{kdv}) 
is given by 
\begin{equation}
\label{cnoidal-wave}
u(x,t) = \phi_0(x-c_0t), \qquad \phi_0(x) := 2k^2 {\rm cn}^2(x,k), 
\quad 
c_0 := 4(2k^2-1),
\end{equation}
where ${\rm cn}(x,k)$ is the Jacobi elliptic function, and $k \in (0,1)$ is the elliptic modulus. Table \ref{Table-1} collects together elliptic integrals and Jacobi elliptic functions used in our work, see \cite{Byrd,Grad,Lawden}.

The main result of this work is the derivation and analysis of two solution families of the KdV equation (\ref{kdv}) parametrized by $\lambda$ and $x_0 \in \mathbb{R}$, where $\lambda$ belongs to $(-\infty,-k^2)$ for the first family 
and $(1-2k^2,1-k^2)$ for the second family. Both the solution families can be expressed in the form 
\begin{align}
u(x,t) = 2 \left[ k^2 - 1 + \frac{E(k)}{K(k)} \right]  + 2
\partial_x^2 \log \tau(x,t), 
\label{new-solution}
\end{align}
where the $\tau$-function for the first family is given by 
\begin{align}
\label{tau-function}
\tau(x,t) := \Theta(x - c_0 t + \alpha_b) 
e^{\kappa_b (x-c_b t + x_0)} + 
\Theta(x - c_0 t - \alpha_b) 
e^{-\kappa_b (x-c_b t + x_0)}
\end{align}
with uniquely defined $\kappa_b > 0$, $c_b > c_0$ and
$\alpha_b \in (0,K(k))$ and the $\tau$-function for the second family
is given by
\begin{align}
\label{tau-function-dark}
\tau(x,t) := \Theta(x - c_0 t + \alpha_d) 
e^{-\kappa_d (x-c_d t + x_0)} + 
\Theta(x - c_0 t - \alpha_d) 
e^{\kappa_d (x-c_d t + x_0)}
\end{align}
with uniquely defined $\kappa_d > 0$, $c_d < c_0$, and
$\alpha_d \in (0,K(k))$.

Figure \ref{bright sol} depicts the spatiotemporal evolution of a
solution $u(x,t)$ given by (\ref{new-solution}) and
(\ref{tau-function}).  This solution represents a bright breather on a
cnoidal wave background (hereafter referred to as a bright breather)
with speed $c_b > c_0$ and inverse width $\kappa_b$, where $c_0$ is the
speed of the background cnoidal wave. As a result of the bright
soliton, the cnoidal wave background is spatially shifted by
$-2 \alpha_b$.

\begin{figure}[h!]
	\centering
	\includegraphics[width=7cm,height=6cm]{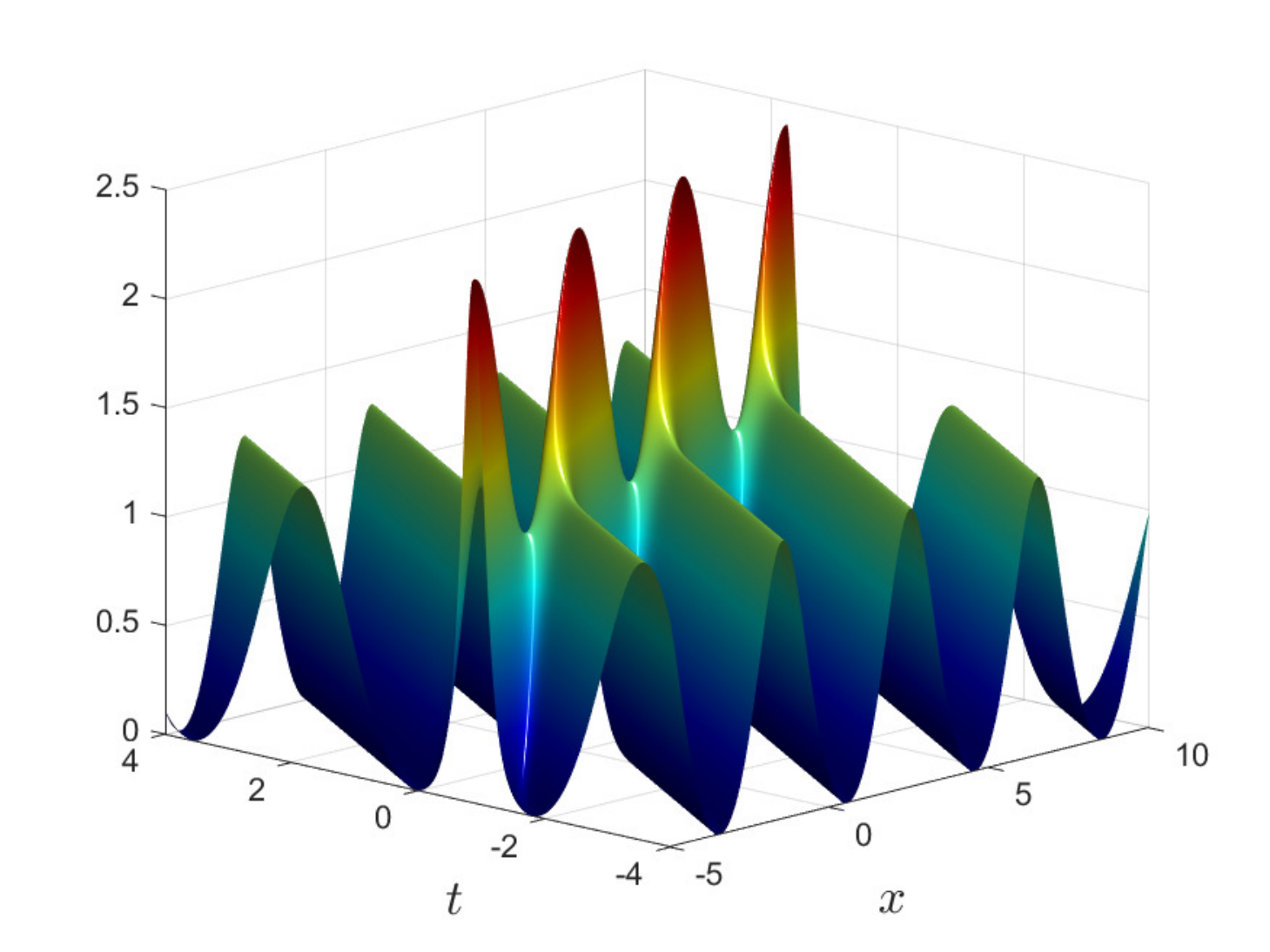}
	\includegraphics[width=7cm,height=6cm]{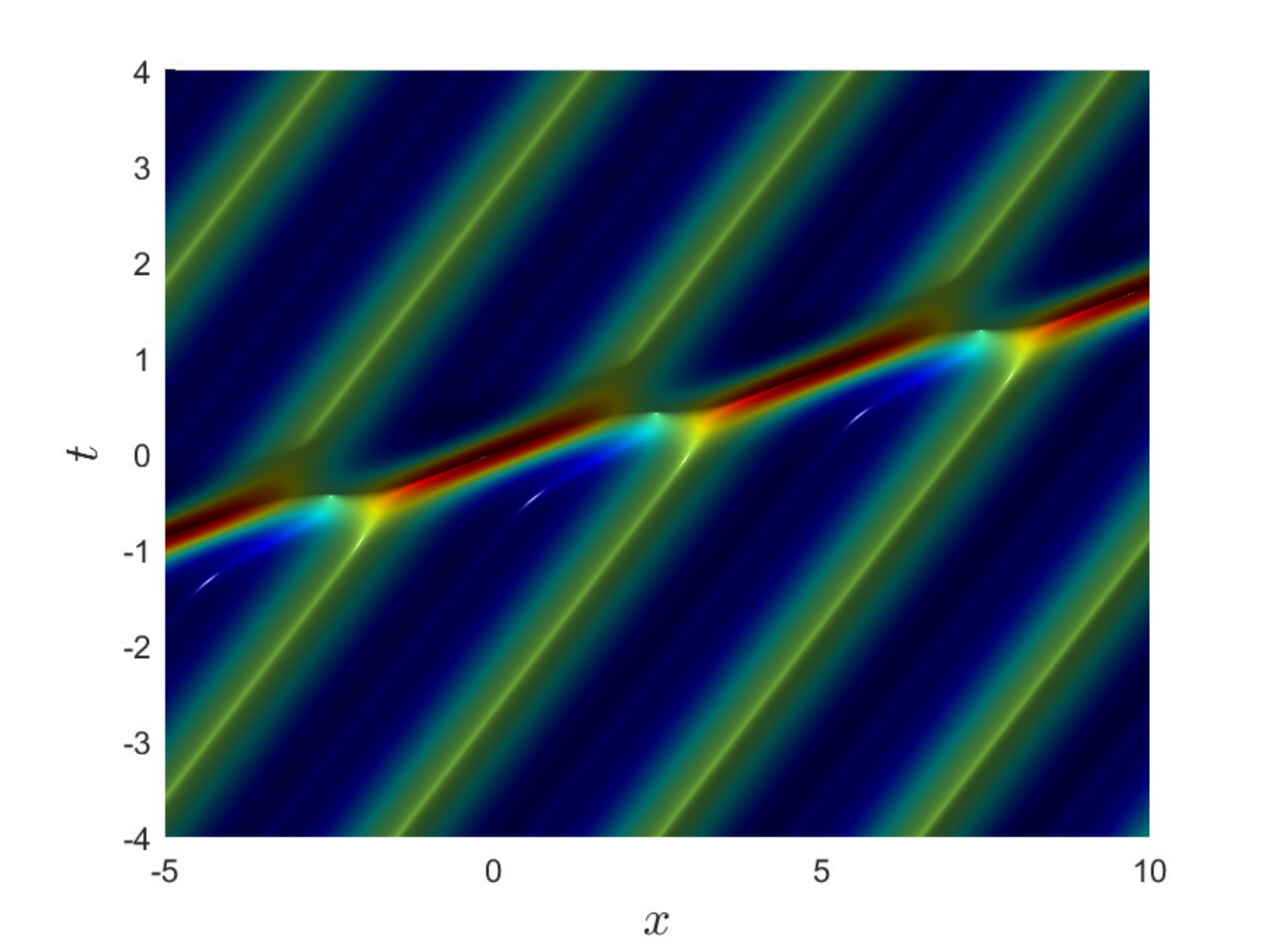}
	\caption{Bright breather on a cnoidal wave with $k=0.8$ 
		for $\lambda=-1.2$ and $x_0 = 0$.}
	\label{bright sol}
\end{figure}

\begin{figure}[h!]
	\centering
	\includegraphics[width=7cm,height=6cm]{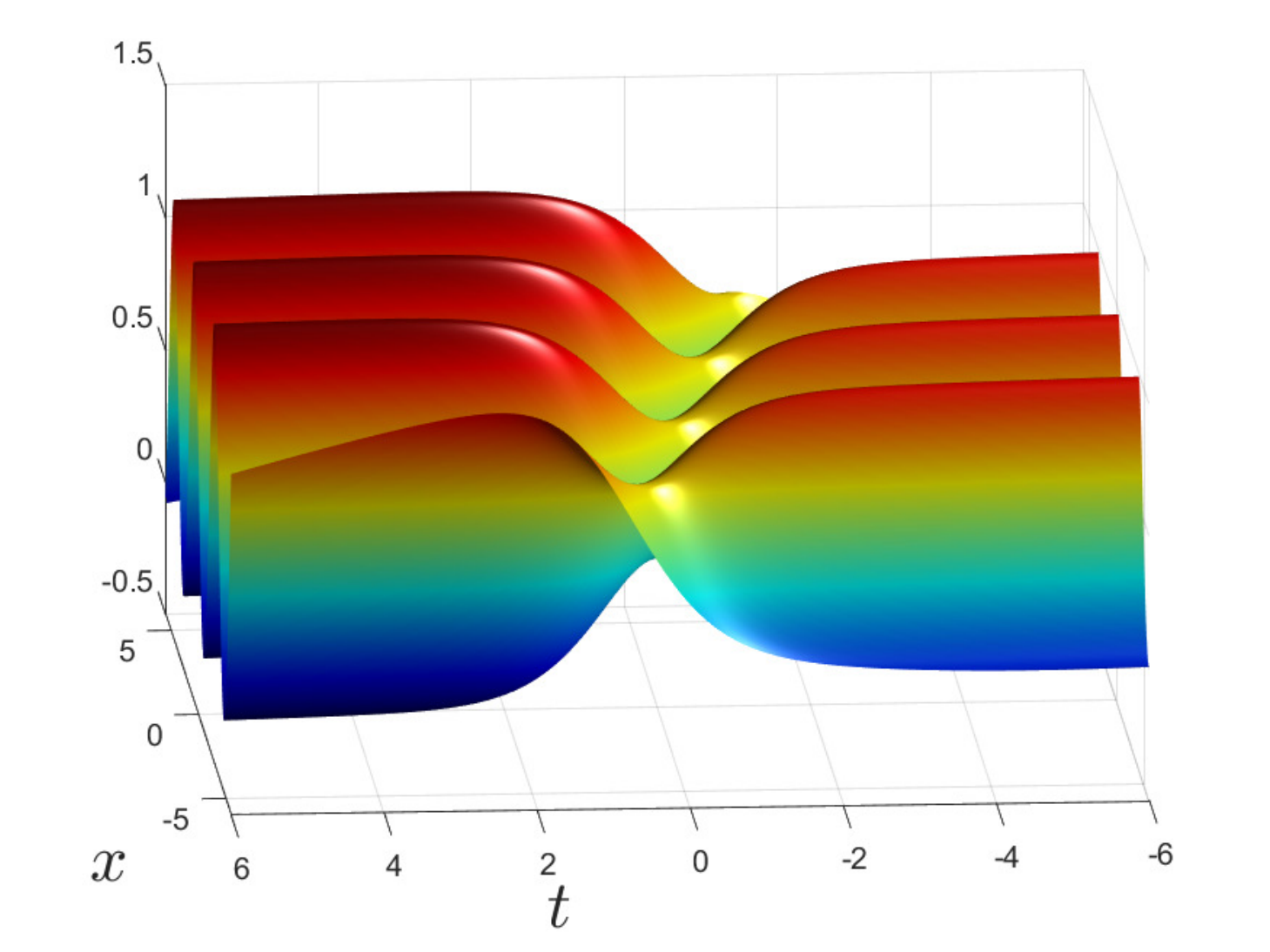}
	\includegraphics[width=7cm,height=6cm]{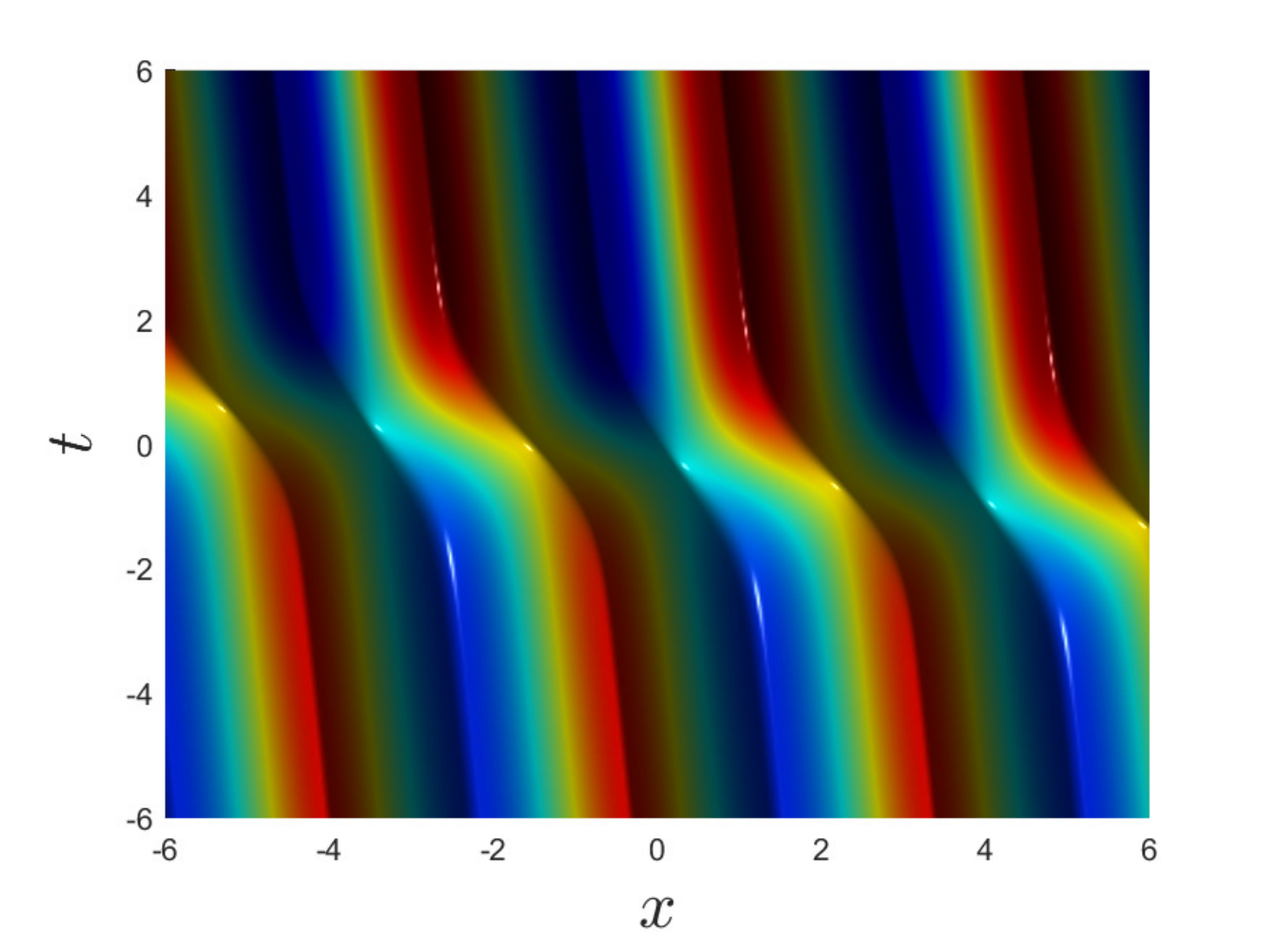}
	\caption{Dark breather on a cnoidal wave with $k=0.7$ for
      $\lambda=0.265$ and $x_0 = 0$.}
	\label{dark sol}
\end{figure}

Figure \ref{dark sol} shows the spatiotemporal evolution of a solution
$u(x,t)$ given by (\ref{new-solution}) and (\ref{tau-function-dark}).
This solution is a dark breather on a cnoidal wave background
(hereafter referred to as a dark breather), where the breather core
exhibits the inverse spatial width $\kappa_d$ and speed $c_d < c_0$.
The dark breather gives rise to the spatial shift $2 \alpha_d$ of the
cnoidal background.

Using properties of Jacobi elliptic functions, we obtain explicit
expressions for the parameters of the $\tau$-functions
(\ref{tau-function}) and (\ref{tau-function-dark}) and their
dependence on the parameter $\lambda$ that characterizes the dynamical
properties of bright and dark breathers. Although the analytical
expressions (\ref{new-solution}) with either (\ref{tau-function}) or
(\ref{tau-function-dark}) are not novel and can be found in equivalent
forms in \cite{Kuznetsov,Nakan1,Bertola}, it is the first time to the
best of our knowledge that the dynamical properties of bright and dark
breathers have been explicitly investigated for the KdV equation (\ref{kdv}). We also obtain asymptotic expressions for bright and dark breathers in the limits
when $\lambda$ approaches the band edges or when the elliptic modulus
$k$ approaches the end points $0$ and $1$.

\section{Traveling cnoidal wave}
\label{sec-2}

A traveling wave solution $u(x,t) = \phi(x-ct)$ to the KdV equation (\ref{kdv}) satisfies the second-order differential equation after integration in $x$:
\begin{equation}
  \label{trav-wave}
  \phi'' + 3 \phi^2 - c \phi = b,
\end{equation}
where $b \in \mathbb{R}$ is the integration constant and the single variable $x$ stands for $x - ct$. The second-order equation (\ref{trav-wave}) is integrable with the first-order invariant 
\begin{equation}
\label{first-order}
(\phi')^2 + 2 \phi^3 - c \phi^2 - 2 b \phi = d,
\end{equation}
where $d \in \mathbb{R}$  is another integration constant. The following proposition summarizes the existence of periodic solutions to system (\ref{trav-wave}) and (\ref{first-order}). 

\begin{proposition}
	\label{prop-periodic}
	There exists a family of periodic solutions to system (\ref{trav-wave}) and (\ref{first-order}) for every $(b,c,d)$ satisfying 
	$c^2 + 12 b > 0$ and $d \in (U(\phi_+),U(\phi_-))$, where 
	$U(\phi) := 2 \phi^3 - c \phi^2 - 2 b \phi$ and $\phi_{\pm}$ are critical points of $U$ given by $\phi_{\pm} = (c \pm \sqrt{c^2 + 12 b})/6$. 
\end{proposition}

\begin{proof}
If $c^2 + 12 b > 0$, the mapping $\phi \mapsto U(\phi)$ has two critical points 
$\phi_{\pm}$. Since $U'(\phi_{\pm}) = 6 \phi_{\pm}^2 - 2 c \phi_{\pm} - 2b = 0$ 
and  $U''(\phi_{\pm}) = 12 \phi_{\pm} - 2c = \pm 2 \sqrt{c^2 + 12b}$, 
$\phi_+$ is the minimum of $U$ and $\phi_-$ is the maximum of $U$. 
If $d = U(\phi_+)$, the only bounded solution of system (\ref{trav-wave}) and (\ref{first-order}) is a constant solution corresponding to the center point 
$(\phi_+,0)$. If $d = U(\phi_-)$, the only bounded solution of system (\ref{trav-wave}) and (\ref{first-order}) is a homoclinic orbit from the saddle point $(\phi_-,0)$ which surrounds the center point $(\phi_+,0)$. The family of periodic orbits exists in a punctured neighbourhood around the
center point enclosed by the homoclinic orbit, for $d \in (U(\phi_+),U(\phi_-))$. 

If $c^2 + 12 b \leq 0$, the mapping  $\phi \mapsto U(\phi)$ is monotonically increasing. There exist no bounded solutions of system (\ref{trav-wave}) and (\ref{first-order}) with the exception of the constant solution $\phi = c/6$ in the marginal case $c^2 + 12 b = 0$.
\end{proof}

It follows from Proposition \ref{prop-periodic} that the most general
periodic traveling wave solution has three parameters $(b,c,d)$, up to
translations, that are defined in a subset of $\mathbb{R}^3$ for which
$c^2 + 12b > 0$ and $d \in (U(\phi_+),U(\phi_-))$. For each $(b,c,d)$
in this subset of $\mathbb{R}^3$, the translational parameter
$x_0 \in \mathbb{R}$ generates the family of solutions $\phi(x + x_0)$
due to translation symmetry.

Two of the three parameters of the periodic solution family can be
chosen arbitrarily due to the following two symmetries of the KdV
equation (\ref{kdv}):
\begin{itemize}
\item Scaling transformation: if $u(x,t)$ is a solution, so is
  $\alpha^2 u(\alpha x, \alpha^3 t)$, $\alpha > 0$.
\item Galilean transformation: if $u(x,t)$ is a solution, so is
  $\beta + u(x-6\beta t,t)$, $\beta \in \mathbb{R}$.
\end{itemize}
Due to these symmetries, if $\phi_0$ is a periodic solution 
to system (\ref{trav-wave}) and (\ref{first-order}) with $(b,c,d) = (b_0,c_0,d_0)$, then $\beta + \alpha^2 \phi_0(\alpha x)$ 
is also a periodic solution to system (\ref{trav-wave}) and (\ref{first-order}) 
with 
$$
(b,c,d) = (-3 \beta^2 -\alpha^2 \beta c_0 + \alpha^4 b_0,6\beta + \alpha^2 c_0,2 \beta^3 + \alpha^2 \beta^2 c_0 - 2 \beta \alpha^4 b_0 + \alpha^6 d_0),
$$
where $\alpha > 0$ and $\beta \in \mathbb{R}$ are arbitrary
parameters.  Thus, without loss of generality, we can consider the
normalized, 1-parameter family of periodic traveling waves
$\phi_0(x) = 2 k^2 {\rm cn}^2(x,k)$ for which the values of
$(b_0,c_0,d_0)$ are determined in the following proposition.

\begin{proposition}
	\label{prop-normalized} The normalized cnoidal wave $\phi_0(x) = 2 k^2 {\rm cn}^2(x,k)$ is a periodic solution of system (\ref{trav-wave}) and (\ref{first-order}) with
	$$
	b_0 := 4k^2(1-k^2), \quad 
	c_0 := 4(2k^2-1), \quad 
	d_0 = 0,
	$$
	where $k \in (0,1)$ is an arbitrary parameter.
\end{proposition}

\begin{proof}
  Since $\min\limits_{x \in \mathbb{R}} \phi_0(x) = 0$, it follows
  from \eqref{first-order} that $d_0 = U(0) = 0$. On the other hand,
  by using the following fundamental relations between Jacobi elliptic
  functions
  \begin{equation}
    \label{fund-relation}
    {\rm sn}^2(x,k) + {\rm cn}^2(x,k) = 1, \quad 
    {\rm dn}^2(x,k) + k^2 {\rm sn}^2(x,k) = 1
  \end{equation}
  and their derivatives
  \begin{equation}
    \label{derivative-relation}
    \frac{d}{dx} \left[ \begin{array}{l} {\rm sn}(x,k) \\
                          {\rm cn}(x,k) \\
                          {\rm dn}(x,k) \end{array} \right] = 
                      \left[ \begin{array}{l} {\rm cn}(x,k) {\rm dn}(x,k) \\
                               - {\rm sn}(x,k) {\rm dn}(x,k) \\ -k^2
                               {\rm sn}(x,k) {\rm cn}(x,k) \end{array}
                           \right], 
   \end{equation}
   we obtain from (\ref{first-order}) with $d_0 = 0$ that
   $b_0 = 4k^2(1-k^2)$ and $c_0 = 4(2k^2-1)$.
\end{proof}

\section{Lam\'{e} equation as the spectral problem}
\label{sec-3}

The spectral problem (\ref{stat-Sch}) with the normalized cnoidal wave
(\ref{cnoidal-wave}) is known as the Lam\'{e} equation
\cite[p.395]{Ince}.  It can be written in the form
\begin{equation}
\label{Lame}
v''(x) - 2 k^2 {\rm sn}^2(x,k) v(x) + \eta v(x) = 0, \quad \eta := \lambda + 2k^2,
\end{equation} 
where the single variable $x$ stands for $x - c_0t$. 
By using (\ref{fund-relation}) and (\ref{derivative-relation}), 
we obtain the following three particular solutions $v = v_{1,2,3}(x)$ of the Lam\'{e} equation (\ref{Lame}) with $\lambda = \lambda_{1,2,3}(k)$:
\begin{align*}
&\lambda_1(k) := -k^2, \qquad & v_1(x) := {\rm dn}(x,k), \\
&\lambda_2(k) := 1-2k^2, \qquad & v_2(x) := {\rm cn}(x,k), \\
&\lambda_3(k) := 1-k^2, \qquad & v_3(x) := {\rm sn}(x,k), 
\end{align*}
which correspond to the three remarkable values of $\eta$:
$\eta_1 = k^2$, $\eta_2 = 1$, and $\eta_3 = 1+k^2$. For $k \in (0,1)$, the three eigenvalues are sorted as $\lambda_1(k) < \lambda_2(k) < \lambda_3(k)$.

\begin{figure}[h]
	\centering
	\includegraphics[width=7cm,height=6cm]{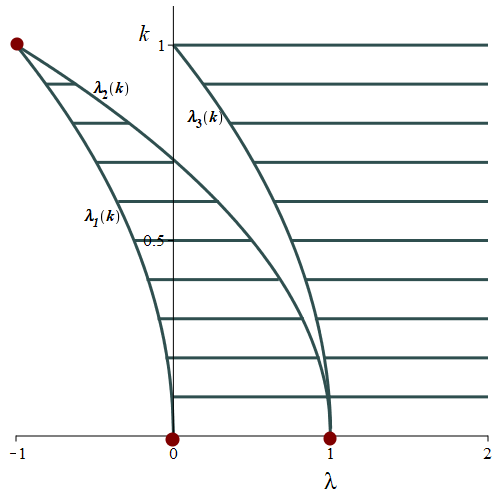}
	\caption{Floquet spectrum of the Lam\'{e} equation (\ref{Lame}) 
		for different values of $k \in (0,1)$.}
	\label{band structure}
\end{figure}

Figure \ref{band structure} shows the Floquet spectrum of the Lam\'{e}
equation (\ref{Lame}), which corresponds to the admissible values of
$\lambda$ for which $v \in L^{\infty}(\mathbb{R})$. The bands are
shaded and the band edges shown by the bold solid curves corresponding
to $\lambda = \lambda_{1,2,3}(k)$ for $k \in (0,1)$.  The cnoidal wave
is the periodic potential with a single finite gap (the so-called {\em
  one-zone} potential) \cite{Oblak} so that the Floquet spectrum
consists of the single finite band $[\lambda_1(k),\lambda_2(k)]$ and
the semi-infinite band $[\lambda_3(k),\infty)$.

As is well-known (see \cite[p.~395]{Ince}), the two linearly
independent solutions of the Lam\'{e} equation (\ref{Lame}) for
$\lambda \neq \lambda_{1,2,3}(k)$ are given by the functions
\begin{equation}
  \label{two-solutions}
  v_{\pm}(x) = \frac{H(x \pm \alpha)}{\Theta(x)} e^{\mp x Z(\alpha)}, 
\end{equation}
where $\alpha \in \mathbb{C}$ is found from $\lambda \in \mathbb{R}$
by using the characteristic equation
$\eta = k^2 + {\rm dn}^2(\alpha,k)$ and the Jacobi zeta function is
$Z(\alpha) := \frac{\Theta'(\alpha)}{\Theta(\alpha)} =
Z(\varphi_\alpha,k)$ with $\varphi_\alpha = \mathrm{am}(\alpha,k)$
\cite[144.01]{Byrd}, see Table \ref{Table-1}. Since
$\eta = \lambda + 2k^2$, the characteristic equation can be written in
the form
\begin{equation}
\label{charact-eq}
\lambda = 1 - 2k^2 + k^2 {\rm cn}^2(\alpha,k).
\end{equation}

\vspace{0.1cm}

The following proposition clarifies how $\alpha$ is defined from the characteristic equation (\ref{charact-eq}) when $\lambda$ is decreased from $\lambda_3(k)$ to $-\infty$. Figure \ref{path alpha} illustrates the path of $\alpha$ in the complex plane.

\begin{figure}[h!]
	\centering
	\includegraphics[width=5cm,height=5cm]{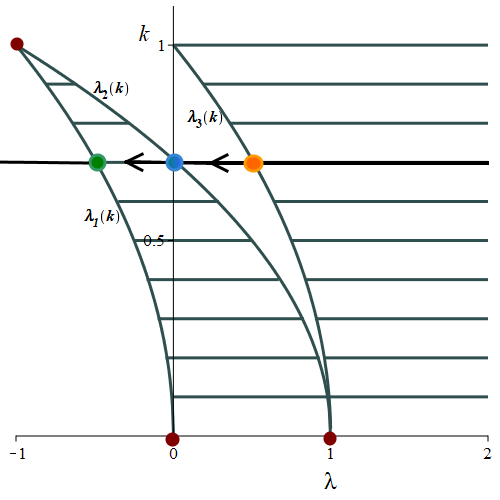}
	\includegraphics[width=8cm,height=5cm]{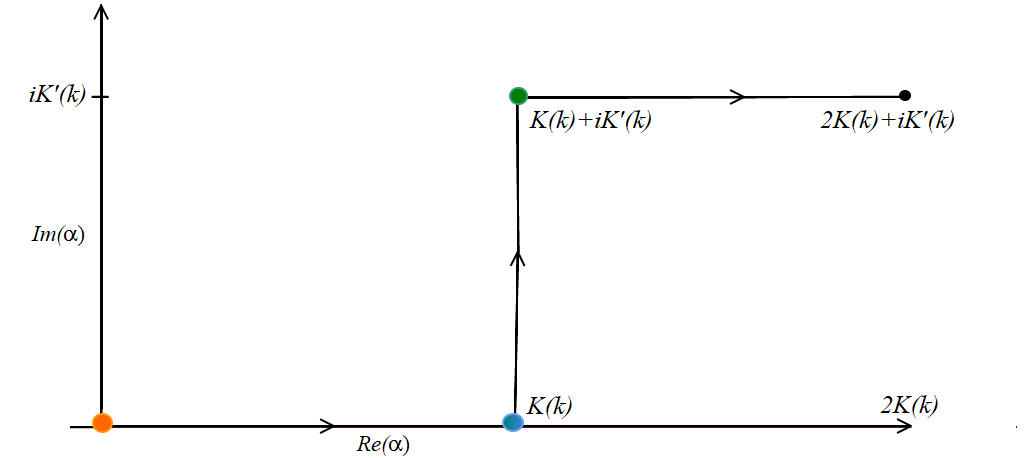}
	\caption{Left: Floquet spectrum with orange, blue, and green dots
      corresponding to $\lambda_3(k)$, $\lambda_2(k)$, and
      $\lambda_1(k)$, respectively, for a fixed value of
      $k \in (0,1)$. Right: The complex plane for the parameter
      $\alpha$ indicating the path of $\alpha$ corresponding to the
      path of $\lambda$ in \eqref{charact-eq}.}
	\label{path alpha}
\end{figure}

\begin{proposition}
	\label{prop-path}
	Fix $k \in (0,1)$. We have 
	\begin{itemize}
		\item $\alpha = F(\varphi_{\alpha},k) \in [0,K(k)]$ for $\lambda \in [\lambda_2(k),\lambda_3(k)]$, where $\varphi_{\alpha} \in [0,\frac{\pi}{2}]$ is given by 
		  \begin{equation}
		\label{char-eq}
\sin \varphi_\alpha = \frac{\sqrt{1-k^2 - \lambda}}{k}.
		\end{equation}
		\item $\alpha = K(k) + i \beta$ with $\beta = F(\varphi_{\beta},k') \in [0,K'(k)]$ for $\lambda \in [\lambda_1(k),\lambda_2(k)]$, where 
		$\varphi_{\beta} \in [0,\frac{\pi}{2}]$ is given by 
		  \begin{equation}
		  \label{eq-beta}
\sin \varphi_\beta =
		\frac{\sqrt{1-2k^2-\lambda}}{\sqrt{(1-k^2)(1-k^2-\lambda)}}.
		\end{equation}
	\item $\alpha = K(k) + i K'(k) + \gamma$ with $\gamma = F(\varphi_{\gamma},k) \in [0,K(k))$ for $\lambda \in (-\infty,\lambda_1(k)]$, where $\varphi_{\gamma} \in [0,\frac{\pi}{2})$ is given by 
        \begin{equation}
\label{eq:2}
\sin \varphi_\gamma =
\frac{\sqrt{-k^2-\lambda}}{\sqrt{1-2k^2-\lambda}},
\end{equation}
	\end{itemize}
where $k' = \sqrt{1-k^2}$ and $K'(k) = K(k')$.
\end{proposition}

\begin{proof}
  When $\lambda \in [\lambda_2(k),\lambda_3(k)]$, it follows from
  (\ref{charact-eq}) that ${\rm cn}^2(\alpha,k) \in [0,1]$ and hence
  $\alpha \in [0,K(k)] \; {\rm mod}\, K(k)$.  Solving
  (\ref{charact-eq}) for  $\sin \varphi_{\alpha} = {\rm sn}(\alpha,k)$ using \eqref{fund-relation} yields
  (\ref{char-eq}). As $\lambda$ is decreased from $\lambda_3(k)$ to
  $\lambda_2(k)$, $\varphi_\alpha$ is monotonically increasing from
  $0$ to $\pi/2$ and so $\alpha = F(\varphi_\alpha,k)$ monotonically
  increases from $0$ to $K(k)$. See the orange and blue dots in Figure
  \ref{path alpha}.

  When $\lambda \in [\lambda_1(k),\lambda_2(k)]$, we use the special
  relations (see \cite[8.151 and 8.153]{Grad}),
  $$
  {\rm cn}(K(k) + i \beta,k) = -k' \frac{{\rm sn}(i \beta,k)}{{\rm
      dn}(i\beta,k)} = -i k' \frac{{\rm sn}(\beta,k')}{{\rm
      dn}(\beta,k')},
  $$
  where $k' := \sqrt{1-k^2}$. The characteristic equation
  (\ref{charact-eq}) is rewritten in the form
  $$
  {\rm sn}^2(\beta,k') = \frac{1 - 2k^2 - \lambda}{(1-k^2)
    (1-k^2-\lambda)},
  $$
  from which it follows that ${\rm sn}^2(\beta,k') \in [0,1]$ and
  hence $\beta \in [0,K(k')] \; {\rm mod} \, K(k')$.  Setting
  $\sin \varphi_{\beta} = {\rm sn}(\beta,k')$ yields
  (\ref{eq-beta}). When $\lambda$ is decreased from $\lambda_2(k)$ to
  $\lambda_1(k)$, then $\varphi_\beta$ is monotone increasing and so
  is $F(\varphi_{\beta},k')$. Hence, $\beta$ increases from $0$ to
  $K'(k)$. See blue and green dots on Figure \ref{path alpha}.

  When $\lambda \in (-\infty,\lambda_1(k)]$, we use the special
  relations (see \cite[8.151]{Grad}),
  $$
  {\rm cn}(K(k) + i K'(k) + \gamma) = -\frac{i k'}{k {\rm
      cn}(\gamma,k)},
  $$
  and rewrite the characteristic equation (\ref{charact-eq}) in the
  form
  \begin{equation*}
    {\rm cn}^2(\gamma,k) = \frac{1-k^2}{1 - 2k^2 - \lambda}.
  \end{equation*}
  from which it follows that ${\rm cn}^2(\gamma,k) \in [0,1]$ and
  hence $\gamma \in [0,K(k)) \; {\rm mod} \, K(k)$. Setting
  $\sin \varphi_{\gamma} = {\rm sn}(\gamma,k)$ and using \eqref{fund-relation} yield
  (\ref{eq:2}). When $\lambda$ is decreased from $\lambda_1(k)$ to
  $-\infty$, then $\varphi_{\gamma}$ is monotone increasing and so is
  $F(\varphi_{\gamma},k)$. Hence, $\gamma$ increases from $0$ to
  $K(k)$. See the green and black dots in Figure \ref{path alpha}.
\end{proof}

\section{Time evolution of the eigenfunctions}
\label{sec-4}

Let $u(x,t) = \phi_0(x-c_0t)$ be the normalized cnoidal wave
(\ref{cnoidal-wave}) and $v(x,t) = v_{\pm}(x,t)$ be solutions of the
system (\ref{stat-Sch}) and (\ref{time-evol}) such that
$v_{\pm}(x,0) = v_{\pm}(x)$ is given by (\ref{two-solutions}).  The
time dependence of $v_{\pm}(x,t)$ can be found by separation of
variables:
\begin{equation}
  \label{eigenfunctions}
  v_{\pm}(x,t) = \frac{H(x - c_0t \pm \alpha)}{\Theta(x-c_0t)} e^{\mp
    (x - c_0t) Z(\alpha) \mp t \omega(\alpha)}, 
\end{equation}
where $\omega(\alpha)$ is to be found. After substituting (\ref{eigenfunctions}) into (\ref{time-evol}) and dividing by $v_{\pm}(x,t)$, we obtain
\begin{equation}
\label{time}
\omega(\alpha) = (c_0 + 4 \lambda - 2 \phi_0(x)) \left[ Z(\alpha) \pm Z(x) 
\mp \frac{H'(x \pm \alpha)}{H(x \pm \alpha)} \right] \mp \phi_0'(x),
\end{equation}
where $x$ stands again for $x-ct$. Equation (\ref{time}) holds for
every $x \in \mathbb{R}$ due to the compatibility of the system
(\ref{stat-Sch}) and (\ref{time-evol}). Hence, we obtain
$\omega(\alpha)$ by substituting $c_0 = 4(2k^2-1)$ and evaluating
(\ref{time}) at $x = 0$:
\begin{equation}
  \label{omega}
  \omega(\alpha) = 4(\lambda + k^2 - 1) \left[
    \frac{\Theta'(\alpha)}{\Theta(\alpha)} 
    - \frac{H'(\alpha)}{H(\alpha)} \right],
\end{equation}
where we have used the parity properties \cite[8.192]{Grad}:
$$
H(-x) = -H(x) \qquad \mbox{\rm and} \qquad \Theta(-x) = \Theta(x).
$$
The following proposition ensures that $\omega(\alpha)$ is real when $\lambda$ is taken either in the semi-infinite gap $(-\infty,\lambda_1(k))$ or in the finite gap $(\lambda_2(k),\lambda_3(k))$. 

\begin{proposition}
	\label{prop-omega}
	Fix $k \in (0,1)$. Then, $\omega(\alpha) \in \mathbb{R}$ if $\lambda \in (-\infty,\lambda_1(k)) \cup (\lambda_2(k),\lambda_3(k))$ and $\omega(\alpha) \in i \mathbb{R}$ if $\lambda \in [\lambda_1(k),\lambda_2(k)]$.
\end{proposition}

\begin{proof}
	We recall the logarithmic derivatives of the Jacobi theta functions \cite[8.199(3)]{Grad}:
\begin{align*}
\frac{H'(x)}{H(x)} &= \frac{\pi}{2 K(k)} \left[ \cot\left(\frac{\pi x}{2 K(k)}\right) + 4 \sin\left(\frac{\pi x}{K(k)}\right) \sum_{n=1}^{\infty} \frac{q^{2n}}{1 - 2 q^{2n} \cos\left(\frac{\pi x}{K(k)}\right) + q^{4n}} \right], \\
\frac{H_1'(x)}{H_1(x)} &= -\frac{\pi}{2 K(k)} \left[  \tan\left(\frac{\pi x}{2 K(k)}\right) + 4 \sin\left(\frac{\pi x}{K(k)}\right) \sum_{n=1}^{\infty} \frac{q^{2n}}{1 + 2 q^{2n} \cos\left(\frac{\pi x}{K(k)}\right) + q^{4n}} \right], \\
\frac{\Theta_1'(x)}{\Theta_1(x)} &= -\frac{2\pi}{K(k)} \sin\left(\frac{\pi x}{K(k)}\right) \sum_{n=1}^{\infty} \frac{q^{2n-1}}{1 + 2 q^{2n} \cos\left(\frac{\pi x}{K(k)}\right) + q^{4n-2}}, \\
\frac{\Theta'(x)}{\Theta(x)} &= \frac{2\pi}{K(k)} \sin\left(\frac{\pi x}{K(k)}\right) \sum_{n=1}^{\infty} \frac{q^{2n-1}}{1 - 2 q^{2n} \cos\left(\frac{\pi x}{K(k)}\right) + q^{4n-2}},
\end{align*}	
where $q := e^{-\frac{\pi K'(k)}{K(k)}}$ is the Jacobi nome, see Table \ref{Table-1}.

If $\lambda \in [\lambda_2(k),\lambda_3(k)]$, then $\alpha = F(\varphi_{\alpha},k) \in [0,K(k)]$ by Proposition \ref{prop-path} and  (\ref{omega}) returns real $\omega(\alpha)$, where both logarithmic 
derivatives of the Jacobi theta functions are positive.

If $\lambda \in [\lambda_1(k),\lambda_2(k)]$, then $\alpha = K(k) + i \beta$ with $\beta = F(\varphi_{\beta},k') \in [0,K'(k)]$ by Proposition \ref{prop-path}. The half-period translations \cite[8.183]{Grad} yield
\begin{align*}
H(K(k) + i \beta) &= H_1(i\beta), \\
\Theta(K(k)+i\beta) &= \Theta_1(i\beta), 
\end{align*}
so that the logarthmic derivatives in (\ref{omega}) are purely imaginary and  $\omega(K(k)+i\beta) \in i \mathbb{R}$.

If  $\lambda \in (-\infty,\lambda_1(k)]$, then 
$\alpha = K(k) + i K'(k) + \gamma$ with $\gamma = F(\varphi_{\gamma},k) \in [0,K(k))$ by Proposition \ref{prop-path}. The half-period translations \cite[8.183]{Grad} yield
\begin{align*}
H(K(k) + i K'(k) + \gamma) &= e^{\frac{\pi K'(k)}{4 K(k)}} e^{-\frac{i \pi \gamma}{2 K(k)}} \Theta_1(\gamma), \\
\Theta(K(k) + i K'(k) + \gamma) &= e^{\frac{\pi K'(k)}{4 K(k)}} e^{-\frac{i \pi \gamma}{2 K(k)}} H_1(\gamma),
\end{align*}
The purely imaginary part of the logarithmic derivatives cancels in
(\ref{omega}) after the transformation and we obtain the real quantity
\begin{equation}
\label{omega-gamma}
\omega(K(k) + i K'(k) + \gamma) = 4(\lambda + k^2 - 1) \left[ \frac{H_1'(\gamma)}{H_1(\gamma)}
- \frac{\Theta_1'(\gamma)}{\Theta_1(\gamma)} \right],
\end{equation}
where both logarithmic derivatives are negative.
\end{proof}

\section{New solutions via the Darboux transformation}
\label{sec-5}

We use the standard tool of the one-fold Darboux transformation for
the KdV equation \cite{Matveev}.  If we fix a value of
$\lambda = \lambda_0$ and obtain a solution $v = v_0(x,t)$ of the
linear equations (\ref{stat-Sch}) and (\ref{time-evol}) associated
with the potential $u = \phi_0(x-c_0t)$ of the KdV equation
(\ref{kdv}), a new solution of the same KdV equation (\ref{kdv}) is
given by
\begin{equation}
\label{hat-u}
\hat{u}(x,t) = \phi_0(x-c_0t) + 2 \partial_x^2 \log v_0(x,t).
\end{equation}
The new solution $\hat{u}(x,t)$ is real and non-singular if and only
if $v_0(x,t) \ne 0$ everywhere in the $(x,t)$ plane. This is true for
$\lambda_0 \in (-\infty,\lambda_1(k))$, which is below the Floquet
spectrum (Figure \ref{band structure}), because Sturm's nodal theorem
implies that $v_{\pm}(x,t),$ given by (\ref{eigenfunctions}), are
sign-definite in $x$ for every $t \in \mathbb{R}$. However, if
$\lambda_0 \in (\lambda_2(k),\lambda_3(k))$ is in the finite gap,
Sturm's nodal theorem implies that $v_{\pm}(x,t)$ have exactly one
zero on the fundamental period of $\phi_0$ for every
$t \in \mathbb{R}$. We will show that this technical obstacle can be
overcome with the translation of the new solution $\hat{u}(x,t)$ with
respect to a half-period in the complex plane of $x$.

The following proposition gives an important relation between the Jacobi cnoidal function and the Jacobi theta function. 

\begin{proposition}
	\label{prop-relation}
	For every $k \in (0,1)$, we have 
	\begin{equation}
	\label{relation}
	k^2 {\rm cn}^2(x,k) = k^2 - 1 + \frac{E(k)}{K(k)} + \partial_x^2
	\log \Theta(x). 
	\end{equation}
\end{proposition}

\begin{proof}
	From \cite[6.6.9]{Lawden} we have that
\begin{equation}
  \label{Byrd-relation}
  P\left(\frac{x}{\sqrt{e_1-e_3}} \right) = c_1 - \partial_z^2 \log \theta_1 \left(\frac{\pi x}{2
      K(k)}\right), 
\end{equation}
where $c_1$ is a specific constant to be determined and $P(z)$ is
Weierstrass' elliptic function that satisfies
\begin{equation*}
[P'(z)]^2 = 4 (P(z)-e_1) (P(z)-e_2) (P(z)-e_3),
\end{equation*}
with three turning points $e_3 < e_2 < e_1$ such that
$e_1 + e_2 + e_3 = 0$.  As is well known (see \cite[8.169]{Grad}),
$P(z)$ is related to the Jacobi elliptic functions by
\begin{align*}
  P\left(\frac{x}{\sqrt{e_1-e_3}} \right) &= e_3 + \frac{e_1 - e_3}{{\rm sn}^2(x,k)} \\
&= e_3 + (e_1-e_3) k^2 {\rm sn}^2(x+iK'(k),k) \\
&= e_3 + (e_2-e_3) {\rm sn}^2(x+iK'(k),k) \\
&= e_2 - (e_2-e_3) {\rm cn}^2(x+iK'(k),k),
\end{align*}
where we have used the property $k {\rm sn}(x+i K'(k),k) = {\rm sn}(x,k)$ \cite[8.151]{Grad}, the definition 
$$
k^2 = \frac{e_2-e_3}{e_1-e_3},
$$ 
and the first relation in (\ref{fund-relation}). Thus, we
obtain, due to the relation (\ref{Byrd-relation}) that
\begin{equation}\label{k_cn}
\begin{split}
k^2 {\rm cn}^2(x,k) &= \frac{e_2 - P\left(\frac{x-i K'(k)}{\sqrt{e_1-e_3} }\right)}{e_1-e_3} \\
&= \frac{e_2-c_1}{e_1-e_3} + \partial_x^2 \log \theta_1 \left(\frac{\pi (x - i K'(k))}{2 K(k)}\right) \\
&= \frac{e_2-c_1}{e_1-e_3} + \partial_x^2 \log \theta_4 \left(\frac{\pi x}{2 K(k)}\right) \\
&= \frac{e_2-c_1}{e_1-e_3} + \partial_x^2 \log \Theta(x),
\end{split}
\end{equation}
where we have used the half-period translation  \cite[8.183]{Grad}:
$$
\theta_1\left(u - \frac{i\pi K'(k)}{2K(k)}\right) = -i e^{\frac{\pi K'(k)}{4 K(k)}} e^{iu} \theta_4(u)
$$
and $\partial_x^2 \log e^{c_2 + c_3 x} = 0$ for every $c_2,c_3 \in \mathbb{C}$.
To find the specific constant $\frac{e_2-c_1}{e_1-e_3}$, we evaluate the relation (\ref{k_cn}) at $x = 0$: 
\begin{align*}
\frac{e_2-c_0}{e_1-e_3} &= k^2 - \frac{\Theta''(0)}{\Theta(0)} \\ 
&= k^2 - 1 + \frac{E(k)}{K(k)},
\end{align*}
where we have used \cite[8.196]{Grad}. This yields (\ref{relation}).
\end{proof}

The following two theorems present the construction of bright and dark
breathers in the form (\ref{new-solution}) with either
(\ref{tau-function}) or (\ref{tau-function-dark}). These two theorems
contribute to the main result of this work.

\begin{theorem}
	\label{th-bright}
	There exists an exact solution to the KdV equation (\ref{kdv}) in the form 
	(\ref{new-solution}) with (\ref{tau-function}), where 
	$x_0 \in \mathbb{R}$ is arbitrary and where $\alpha_b \in (0,K(k))$, $\kappa_b > 0$, and $c_b > c_0$ are uniquely defined from $\lambda \in (-\infty,\lambda_1(k))$  by  
	\begin{align}
	\label{eq:3}
	\alpha_b &= F(\varphi_\gamma,k),  \\
	\label{eq:18}
\kappa_b &=
\frac{\sqrt{1-\lambda-k^2} \sqrt{-\lambda-k^2}}{\sqrt{1-2k^2-\lambda}}  - Z(\varphi_\gamma,k), \\
	\label{eq:4}
	c_b &= c_0 + 
	\frac{4 \sqrt{1-\lambda-2k^2} \sqrt{1-\lambda-k^2} \sqrt{-\lambda - k^2}}{\kappa_b},
	\end{align}
with $\varphi_{\gamma} \in (0,\frac{\pi}{2})$ being found from 
\begin{equation}
\label{phi-gamma}
\sin \varphi_\gamma =
\frac{\sqrt{-\lambda-k^2}}{\sqrt{1-2k^2-\lambda}}.
\end{equation}
\end{theorem}

\begin{proof}
Consider a linear combination of the two solutions to the linear system 
(\ref{stat-Sch}) and (\ref{time-evol}) in the form (\ref{eigenfunctions}) with $\alpha = K(k) + i K'(k) + \gamma$ and $\gamma = F(\varphi_{\gamma},k) \in (0,K(k))$:
\begin{equation}
\label{superposition}
v_0(x,t) = c_+ \frac{H(x - c_0t + \alpha)}{\Theta(x-c_0t)} e^{- (x - c_0 t) Z(\alpha) - \omega(\alpha) t} + c_- \frac{H(x - c_0t - \alpha)}{\Theta(x-c_0t)} e^{+(x - c_0 t) Z(\alpha) + \omega(\alpha) t},
\end{equation}
where $(c_+,c_-)$ are arbitary constants. By using the half-period translations 
of the Jacobi theta functions \cite[8.183]{Grad}, we obtain for $\alpha = K(k) + i K'(k) + \gamma$:
\begin{align*}
H(x + \alpha) &= e^{\frac{\pi K'(k)}{4 K(k)} - \frac{i \pi (x + \gamma)}{2 K(k)}} \Theta(x + K(k) + \gamma), \\
H(x - \alpha) &= -e^{\frac{\pi K'(k)}{4 K(k)} + \frac{i \pi (x - \gamma)}{2 K(k)}} \Theta(x + K(k) - \gamma), 
\end{align*}
and
\begin{align*}
Z(\alpha) = \frac{H_1'(\gamma)}{H_1(\gamma)} - \frac{i \pi}{2 K(k)}.
\end{align*}
Substituting these expressions into (\ref{superposition}) cancels the $x$-dependent complex phases. Anticipating (\ref{hat-u}), we set
$$
c_+ = c e^{-(K(k) + x_0) \frac{H_1'(\gamma)}{H_1(\gamma)}}, \quad 
c_- = -c e^{(K(k) + x_0) \frac{H_1'(\gamma)}{H_1(\gamma)}}
$$ 
with arbitrary parameters $c,x_0 \in \mathbb{R}$, from which the
constant $c$ cancels out due to the second logarithmic derivative.
Using $c_\pm$ in \eqref{superposition}, inserting $v_0$ into
\eqref{hat-u}, and simplifying with the help of (\ref{relation}), we obtain a new solution in the final form $u(x,t) := \hat{u}(x-K(k),t)$, where $u(x,t)$ is given by (\ref{new-solution}) with $\tau(x,t)$ given by  (\ref{tau-function})  
with the following parameters: $\alpha_b := \gamma \in (0,K(k))$, 
$\kappa_b := -\frac{H_1'(\gamma)}{H_1(\gamma)} > 0$, and 
\begin{align*}
c_b &:= c_0 - \omega(K(k) + i K'(k) + \gamma) \frac{H_1(\gamma)}{H_1'(\gamma)} \\
&= 4(k^2-\lambda) + 4 (\lambda + k^2 - 1) \frac{\Theta_1'(\gamma)
	H_1(\gamma)}{\Theta_1(\gamma) H_1'(\gamma)}, 
\end{align*}
where we have used (\ref{omega-gamma}). By using the following identities \cite[1053.02]{Byrd}
\begin{align*}
\frac{H_1'(\gamma)}{H_1(\gamma)} &= -\frac{{\rm sn}(\gamma,k) {\rm dn}(\gamma,k)}{{\rm cn}(\gamma,k)} + Z(\gamma), \\
\frac{\Theta_1'(\gamma)}{\Theta_1(\gamma)} &= -
\frac{k^2 {\rm sn}(\gamma,k) {\rm cn}(\gamma,k)}{{\rm dn}(\gamma,k)}
+ Z(\gamma),
\end{align*}
and the relation formulas $Z(\gamma) = Z(\varphi_{\gamma},k)$,
\begin{align*}
{\rm sn}(\gamma,k) = \sin(\varphi_{\gamma}) = \frac{\sqrt{-\lambda-k^2}}{\sqrt{1-2k^2-\lambda}}, \quad 
{\rm cn}(\gamma,k) = \cos(\varphi_{\gamma}) = \frac{\sqrt{1-k^2}}{\sqrt{1-2k^2-\lambda}}, 
\end{align*}
and
\begin{align*}
{\rm dn}(\gamma,k) = \frac{\sqrt{1-k^2} \sqrt{1-\lambda-k^2}}{\sqrt{1-2k^2-\lambda}},
\end{align*}
we express the parameters $\alpha_b$, $\kappa_b$, and $c_b$ in terms
of incomplete elliptic integrals in (\ref{eq:3}), (\ref{eq:18}), and
(\ref{eq:4}).  Since $\kappa_b > 0$, it follows that $c_b > c_0$.
\end{proof}

\begin{remark}
  The solution $u(x,t)$ obtained in the proof of Theorem
  \ref{th-bright} is the half-period translation along the real axis
  of the solution $\hat{u}(x,t)$ defined by (\ref{hat-u}).
\end{remark}

\begin{remark}
  \label{rem-bright}
  Since $\kappa_b > 0$, it follows from (\ref{new-solution}),
  (\ref{tau-function}), and (\ref{relation}) that
  \begin{equation*}
    u(x,t) \to 2 k^2 {\rm cn}^2(x - c_0t \pm \alpha_b,k) \quad \mbox{\rm as} 
    \;\; x-c_b t \to \pm \infty. 
  \end{equation*}
  A suitably normalized phase shift of the background cnoidal wave can
  be written in the form:
  \begin{equation*}
    \Delta_b := -\frac{2\pi \alpha_b}{K(k)} =
    -\frac{2\pi F(\varphi_\gamma,k)}{K(k)} \in (-2\pi,0).
  \end{equation*}
  When $\Delta_b \in (-\pi,0)$, the normalized phase shift is negative.
  When $\Delta_b \in (-2\pi,-\pi]$, the normalized phase shift is
  considered to be positive by a period translation to
  $\Delta_b + 2\pi \in (0,\pi]$.
\end{remark}

\begin{theorem}
	\label{th-dark}
	There exists an exact solution to the KdV equation (\ref{kdv}) in the form 
	(\ref{new-solution}) with (\ref{tau-function-dark}), where 
	$x_0 \in \mathbb{R}$ is arbitrary and where $\alpha_d \in (0,K(k))$, $\kappa_d > 0$, and $c_d < c_0$ are uniquely defined from $\lambda \in (\lambda_2(k),\lambda_3(k))$ by  
	\begin{align}
	\label{eq:6}
          \alpha_d &= F(\varphi_\alpha,k), \\
	\label{eq:25}
	\kappa_d &= Z(\varphi_\alpha,k), \\
	\label{eq:24}
	c_d &= c_0 - \frac{4 \sqrt{(k^2+\lambda)(\lambda-1+2k^2)(1-k^2-\lambda)}}{\kappa_d},
	\end{align}
with $\varphi_{\alpha} \in (0,\frac{\pi}{2})$ being found from 
\begin{equation}
\label{phi-alpha}
\sin \varphi_\alpha =
\frac{\sqrt{1-k^2-\lambda}}{k}.
\end{equation}
\end{theorem}

\begin{proof}
  When $\lambda \in (\lambda_2(k),\lambda_3(k))$,
  $\alpha = F(\varphi_{\alpha},k) \in (0,K(k))$, $\omega(\alpha)$ and
  $Z(\alpha) = Z(\varphi_\alpha,k)$ are real by Propositions
  \ref{prop-path} and \ref{prop-omega}. However, the functions
  $H(x \pm \alpha)$ change sign so that we should express them in
  terms of the functions $\Theta(x \pm \alpha)$ after complex
  translation of phases. This is achieved by the half-period
  translations \cite[8.183]{Grad}:
  \begin{align*}
	H(x + \alpha) &= i e^{-\frac{\pi K'(k)}{4 K(k)} - \frac{i \pi (x +
                    \alpha)}{2 K(k)}} \Theta(x + \alpha - i K'(k)), \\ 
	H(x - \alpha) &= i e^{-\frac{\pi K'(k)}{4 K(k)} - \frac{i \pi (x -
                    \alpha)}{2 K(k)}} \Theta(x - \alpha - i K'(k)). 
  \end{align*}
  The $x$-dependent complex phase is now a multiplier in the linear
  superposition (\ref{superposition}) which vanishes in the result
  due to the second logarithmic derivative. By using (\ref{hat-u}) and
  (\ref{relation}), we set
  $$
  c_+ = c e^{-(x_0-iK'(k)) Z(\alpha)+\frac{i \pi \alpha}{2 K(k)}}, \quad 
  c_- = c e^{(x_0-iK'(k)) Z(\alpha)-\frac{i \pi \alpha}{2 K(k)}},
  $$ 
  and obtain a new solution in the final form
  $u(x,t) := \hat{u}(x+iK'(k),t)$ with the same $u(x,t)$ as in
  (\ref{new-solution}) and with $\tau(x,t)$ given by
  (\ref{tau-function-dark}) with the following parameters:
  $\alpha_d := \alpha \in (0,K(k))$, $\kappa_b := Z(\alpha) > 0$, and
  \begin{align*}
	c_d &= c_0 - \frac{\omega(\alpha)}{Z(\alpha)} \\
        &= 4(k^2-\lambda) + 4 (\lambda + k^2 - 1) \frac{\Theta(\alpha)
          H'(\alpha)}{\Theta'(\alpha) H(\alpha)},
  \end{align*}
  where we have used (\ref{omega}).  Using the following identities
  \cite[1053.02]{Byrd}
  \begin{align*}
  \frac{H'(\alpha)}{H(\alpha)} &= \frac{{\rm
        cn}(\alpha,k) {\rm dn}(\alpha,k)}{{\rm sn}(\alpha,k)} +
                                                 Z(\alpha), \\
    \frac{\Theta'(\alpha)}{\Theta(\alpha)} &= Z(\alpha),
  \end{align*}
  and the relations $Z(\alpha) = Z(\varphi_{\alpha},k)$,
  \begin{align*} {\rm sn}(\alpha,k) = \sin(\varphi_{\alpha}) =
    \frac{\sqrt{1-\lambda-k^2}}{k}, \quad {\rm cn}(\alpha,k) =
    \cos(\varphi_{\alpha}) = \frac{\sqrt{\lambda -1+2k^2}}{k},
  \end{align*}
  and ${\rm dn}(\alpha,k) = \lambda + k^2$, we express the parameters
  $\alpha_d$, $\kappa_d$, and $c_d$ in terms of incomplete elliptic
  integrals as (\ref{eq:6}), (\ref{eq:25}), and (\ref{eq:24}). Since
  $\kappa_d > 0$, we have $c_d < c_0$.
\end{proof}

\begin{remark}
  The solution $u(x,t)$ obtained in the proof of Theorem \ref{th-dark}
  is the half-period translation along the imaginary axis of the
  solution $\hat{u}(x,t)$ defined by (\ref{hat-u}).
\end{remark}

\begin{remark}
  \label{rem-dark}
  Since $Z(\varphi_\alpha,k) > 0$, it follows from
  (\ref{new-solution}), (\ref{tau-function-dark}), and
  (\ref{relation}) that
  \begin{equation*}
    u(x,t) \to 2 k^2 {\rm cn}^2(x - c_0t \mp \alpha_d,k) \quad \mbox{\rm as}
    \;\; x-c_d t \to \pm \infty.
  \end{equation*}
  A suitably normalized phase shift of the background cnoidal wave can
  be written in the form:
  \begin{equation}
    \label{dark-breather-phase-shift}
    \Delta_d = \frac{2\pi \alpha_d}{K(k)} = \frac{\pi
      F(\varphi_\alpha,k)}{K(k)} \in (0,2\pi) .
  \end{equation}
  When $\Delta_d \in (0,\pi]$, the normalized phase shift is
  positive.  When $\Delta_d \in (\pi,2\pi)$, the normalized phase
  shift is considered to be negative by translation to
  $\Delta_d - 2\pi \in (-\pi,0)$.
\end{remark}

\section{Properties of the bright breather}
\label{sec-6}

Figure \ref{dep bright} plots $\Delta_b$, $\kappa_b$, and $c_b$ for a
bright breather as a function of the parameter $\lambda$, see Theorem
\ref{th-bright} and Remark \ref{rem-bright}.  The phase shift
$\Delta_b$ increases monotonically while the inverse width $\kappa_b$
and the breather speed $c_b$ decrease monotonically as $\lambda$
increases from $-\infty$ towards the band edge $\lambda_1(k),$ shown
by the vertical dashed line. Since $c_0 = 1.12$ for $k = 0.8$, we
confirm that $c_b > c_0$, which can also be observed in Figure
\ref{bright sol}.

\begin{figure}[htb!]
  \centering
  \includegraphics{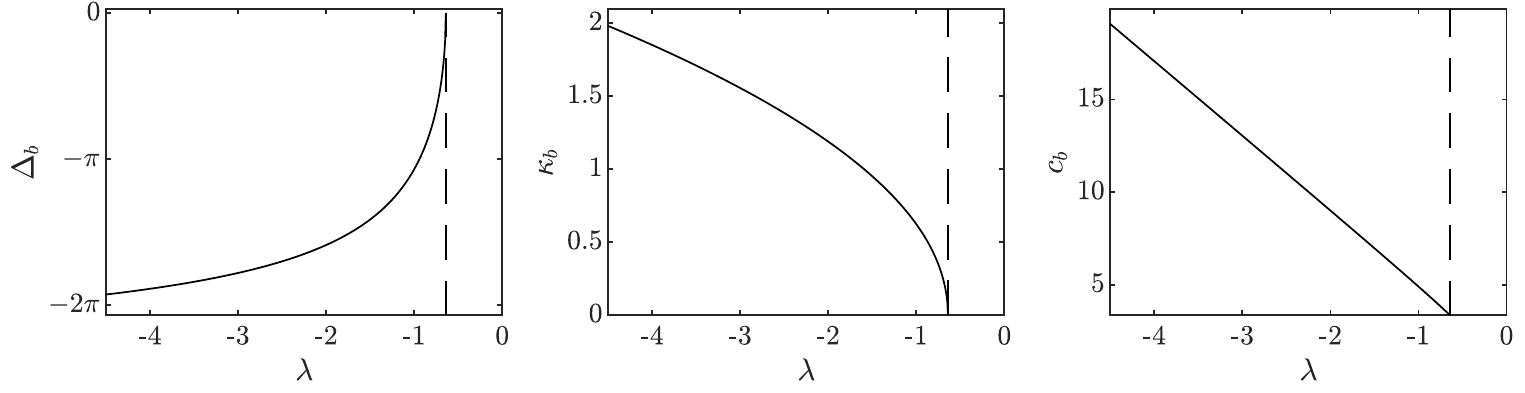}
  \caption{Normalized phase shift $\Delta_b$ (left), inverse
    width $\kappa_b$ (middle), and breather speed $c_b$ (right) versus
    $\lambda$ in $(-\infty,\lambda_1(k))$ for $k=0.8$. The band edge
    $\lambda_1(k)=-k^2$ is shown by the vertical dashed line.}
  \label{dep bright}
\end{figure}

\begin{figure}[htb!]
	\centering
	\includegraphics{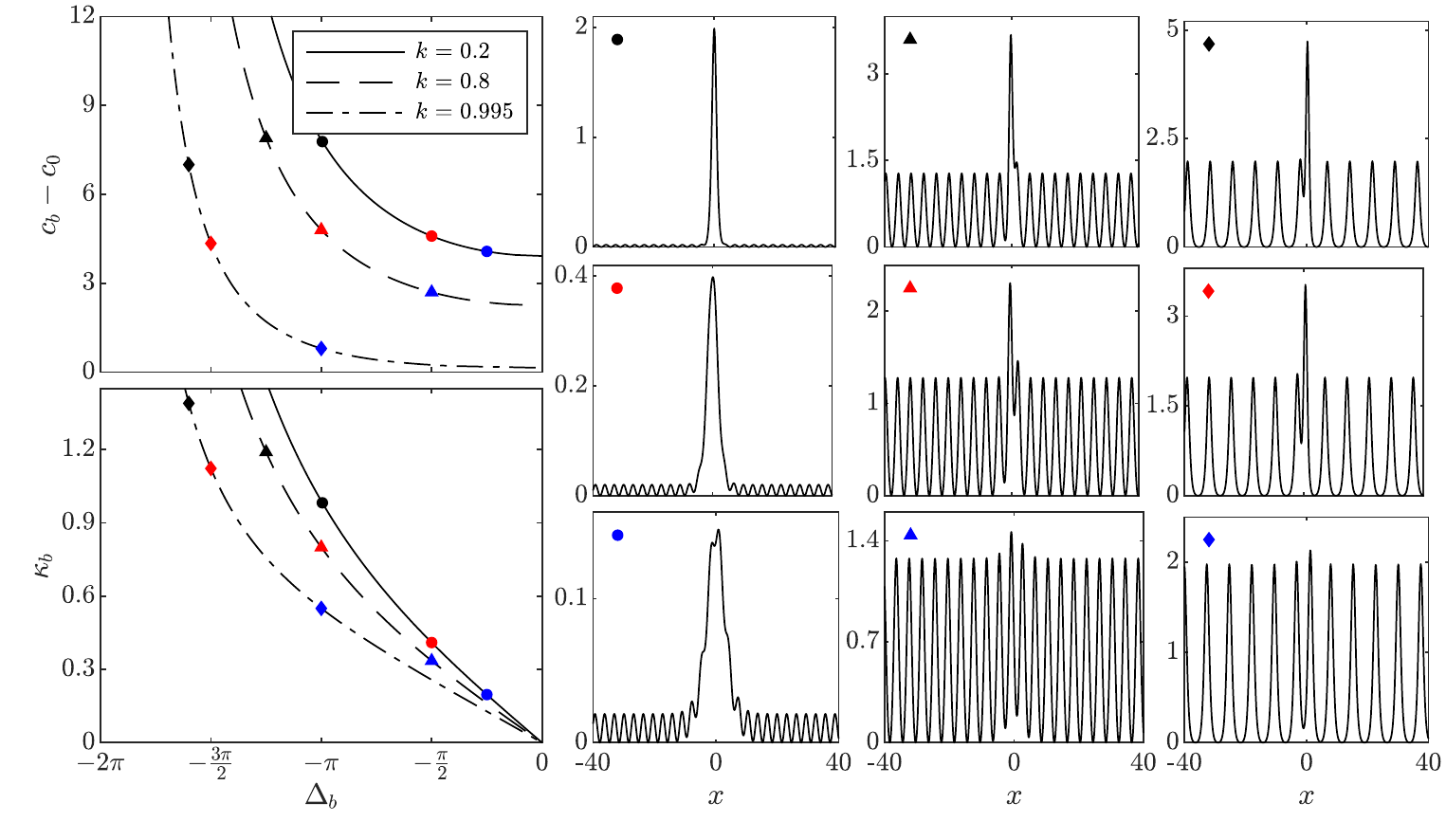}
	\caption{Left top (bottom): dependence of $c_b - c_0$ ($\kappa_b$)
      versus $\Delta_b$ for several values of $k$. Right:
      representative bright breather solutions. Representative
      solutions are marked on the left panel with a unique colored
      symbol.}
	\label{fig:bright-breather-dispersion-modes}
\end{figure}

Figure \ref{fig:bright-breather-dispersion-modes} characterizes the
family of bright breathers by plotting $c_b - c_0$ and $\kappa_b$
versus $\Delta_b$ for three values of $k$. Profiles of representative
breather solutions shown in Figure
\ref{fig:bright-breather-dispersion-modes} confirm why we call them
bright breathers. Bright breathers are more localized, have larger
amplitudes, and move faster for smaller (more negative) values of
$\Delta_b$ (smaller values of $\lambda$).  For sufficiently large
amplitude, $\Delta_b$ falls below $-\pi$ and the breather exhibits a
positive phase shift $\Delta_{b} + 2\pi \in (0,\pi]$ (cf.~Remark
\ref{rem-bright}). In contrast, for sufficiently small-amplitude
breathers, $\Delta_b \in (-\pi,0)$ and the phase shift is negative.

\subsection{Asymptotic limits $\lambda \to -\infty$ and
$\lambda \to \lambda_1(k)$} \hfill\\
\\
It follows from (\ref{phi-gamma}) that
\begin{equation*}
\varphi_{\gamma} = \left\{ 
\begin{array}{ll} \displaystyle
\frac{\pi}{2} - \frac{\sqrt{1-k^2}}{\sqrt{|\lambda|}} + \mathcal{O}(\sqrt{|\lambda|^{-3}}) \quad & \mbox{\rm as} \;\; \lambda \to -\infty, \\  \displaystyle
\frac{\sqrt{|\lambda| - k^2}}{\sqrt{1-k^2}} + \mathcal{O}(\sqrt{(|\lambda| - k^2)^3}) \quad & \mbox{\rm as} \;\; \lambda \to \lambda_1(k).
\end{array}
\right.
\end{equation*}
We also use the following asymptotic expansions of the elliptic integrals:
$$
F(\varphi,k) = \varphi + \mathcal{O}(\varphi^3), \quad 
E(\varphi,k) = \varphi + \mathcal{O}(\varphi^3), \quad \mbox{\rm as} \;\; 
\varphi \to 0
$$
and
$$
F(\varphi,k) = K(k) + \mathcal{O}(\tfrac12\pi-\varphi), \quad 
E(\varphi,k) = E(k) + \mathcal{O}(\tfrac12\pi-\varphi), \quad
\mbox{\rm as} \;\;  
\varphi \to \frac{\pi}{2}
$$
The itemized list below summarizes the asymptotic results, where we use the asymptotic equivalence for the leading-order terms and neglect writing the remainder terms. 

\begin{itemize}
\item The asymptotic values of the normalized phase shift $\Delta_b$ are
\begin{equation*}
\Delta_b \sim \left\{ 
\begin{array}{ll} \displaystyle
-2\pi + \frac{2\pi}{\sqrt{|\lambda|} K(k)} \quad & \mbox{\rm as} \;\; \lambda \to -\infty, \\  \displaystyle
-\frac{2\pi \sqrt{|\lambda|-k^2}}{\sqrt{1-k^2} K(k)} \quad & \mbox{\rm as} \;\; \lambda \to \lambda_1(k).
\end{array}
\right.
\end{equation*}
Since
$\partial_{\varphi} F(\varphi,k) = (1-k^2 \sin^2 \varphi)^{-1/2} > 0$
and $\partial_{\lambda} \varphi_{\gamma} < 0$, the normalized phase
shift $\Delta_b$ is a monotonically increasing function of $\lambda$
from $-2\pi$ to $0$.  This proves that the map $\lambda
\mapsto \Delta_b(\lambda)$ is one-to-one and onto from
$(-\infty,\lambda_1)$ to $(-2\pi,0)$.  \\

\item The asymptotic values for the inverse width $\kappa_b$ are
\begin{equation*}
\kappa_b \sim \left\{ 
\begin{array}{ll} \displaystyle
\sqrt{|\lambda|} \quad & \mbox{\rm as} \;\; \lambda \to -\infty, \\  \displaystyle
\sqrt{\frac{|\lambda|-k^2}{1-k^2}}
\frac{E(k)}{K(k)} \quad & \mbox{\rm as} \;\; \lambda \to \lambda_1(k).
\end{array}
\right.
\end{equation*}
The derivative is given by 
\begin{align*}
  \partial_\lambda \kappa_b &= - \frac{\sin
  \varphi_\gamma}{2\sqrt{1-\lambda-k^2}} \\
& \quad + \Bigg (\sqrt{1-\lambda-k^2}\cos\varphi_\gamma-\sqrt{1-k^2
  \sin^2\varphi_\gamma} + \frac{E(k)}{K(k)\sqrt{1-k^2
  \sin^2\varphi_\gamma}} \Bigg ) \partial_\lambda
  \varphi_\gamma.
\end{align*}
Since the terms in parentheses are positive and $\partial_\lambda \varphi_\gamma < 0$, we have $\partial_\lambda \kappa_b < 0$
so that $\kappa_b$ is a monotonically decreasing
function of $\lambda$.\\ 
  
\item The asymptotic values for the breather speed $c_b$ are 
\begin{equation*}
c_b \sim \left\{ 
\begin{array}{ll} \displaystyle
\frac{4|\lambda|}{2-k^2- E(k)/K(k)} \quad & \mbox{\rm as} \;\; \lambda \to -\infty, \\  \displaystyle
c_0 + 4(1-k^2) \frac{K(k)}{E(k)} \quad & \mbox{\rm as} \;\; \lambda \to \lambda_1(k).
\end{array}
\right.
\end{equation*}
The breather speed $c_b$ in \eqref{eq:4} satisfies $c_b > c_0$. Based on 
the graphs in Fig.~\ref{dep bright}, we conjecture that the breather velocity $c_b$ is a decreasing function of $\lambda$.  
\end{itemize}

\subsection{Asymptotic limits $k \to 0$ and $k \to 1$}
\label{sec:harm-solit-limits} \hfill\\
\\
In the limit $k \to 0$, the background cnoidal wave $\phi_0(x) = 2k^2 {\rm cn}(x;k)$ vanishes since $\Theta(x) \to 1$ as $k \to 0$ whereas 
it follows from (\ref{eq:18}) and (\ref{eq:4}) that 
$$
\kappa_b \to \sqrt{|\lambda|}, \qquad c_b \to 4 |\lambda|,
$$
since $Z(\varphi_{\gamma},k) \to 0$ and $c_0 \to -4$ as $k \to 0$.
The breather solution (\ref{new-solution}) with (\ref{tau-function}) recovers the one-soliton solution
\begin{equation*}
  u(x,t) \to 2 |\lambda| \; \mathrm{sech}^2 \left (
    \sqrt{|\lambda|}(x - 4 |\lambda| t + x_0) \right ) , \quad k \to 0,  
\end{equation*}
for every $\lambda \in (-\infty,0)$. \\
	
In the limit $k \to 1$, the background cnoidal wave
$\phi_0(x) = 2 k^2 {\rm cn}(x;k)$ transforms into the normalized
soliton $\phi_0(x) \to 2\, {\rm sech}^2(x)$ and we will show that the
breather solution (\ref{new-solution}) with (\ref{tau-function})
recovers the two-soliton solution.  It follows from (\ref{eq:18}) and
(\ref{eq:4}) that
$$
\kappa_b \to \sqrt{|\lambda|}, \qquad c_b \to 4 |\lambda|,
$$
since $Z(\varphi_{\gamma},k) \to 0$ and $c_0 \to 4$ as $k \to 1$.
Furthermore, it follows from (\ref{phi-gamma}) that $\varphi_{\gamma} \to \frac{\pi}{2}$ as $k \to 1$ so that $\alpha_b = F(\varphi_{\gamma},k) \to \infty$ as $k\to 1$. In order to regularize the solution, we use the 
translation invariance of the KdV equation, the $2K(k)$-periodicity of
$\Theta$, and define the half-period translation of
\eqref{tau-function} with the transformation $x \to x - K(k)$,
$x_0 \to x_0 + K(k)$:
\begin{equation}
\label{tau-2soli-bright}
\tau(x,t) = \Theta(x - c_0 t + \alpha_b - K(k)) e^{\kappa_b (x-c_b t
	+ x_0)} + \Theta(x - c_0 t - \alpha_b + K(k))
e^{-\kappa_b (x-c_b t + x_0)}.
\end{equation}
Recalling that $\alpha_b = F(\varphi_\gamma,k)$, for each
$\lambda \in (-\infty,-1)$, let us define the phase parameter
$\delta_b$ by evaluating the limit \cite[eq.~(2.14)]{VandeVel}:
\begin{equation}
\label{delta-2soli-bright}
\delta_b := \lim_{k\to 1} \left[ K(k) - F(\varphi_\gamma,k) \right] = \frac{1}{2}
\log \left ( \frac{\sqrt{-\lambda} + 1}{\sqrt{-\lambda} - 1}
\right).
\end{equation}
It remains to deduce the asymptotic formula for $\Theta$ as $k \to 1$. 
We show that 
\begin{equation}
\label{theta4-asymptotics}
\Theta(x) \sim \sqrt{\frac{-2 k' \log
		k'}{\pi}} \cosh(x), \qquad \mbox{\rm as} \;\; k \to 1,
\end{equation}
by using the Poisson summation formula \cite{Boyd}:
\begin{equation}
  \label{eq:9}
  \Theta(x) = \sum_{n=-\infty}^\infty f(n) =
  \sum_{n=-\infty}^\infty\hat{f}(n), 
\end{equation}
where $\hat{f}(m) = \int_{-\infty}^\infty f(n)e^{-2\pi i nm}dn$.
Since
\begin{align*}
  \Theta(x) = 1 + 2 \sum_{n=1}^{\infty} (-1)^n q^{n^2}
  \cos\left(\frac{n \pi x}{K(k)}\right), \qquad q := e^{-\frac{\pi K(k')}{K(k)}},
\end{align*}
where $k' = \sqrt{1-k^2}$, we obtain from (\ref{eq:9}) that
\begin{equation}
  \label{eq:7}
  f(n) = q^{n^2} e^{in\pi \left (1 + x/K(k) \right )}, \quad \hat{f}(n) =
  \sqrt{\frac{K(k)}{K(k')}} (q')^{(n-1/2-x/2K(k))^2} , 
\end{equation}
where $q' := e^{-\frac{\pi K(k)}{K(k')}}$.  As $k \to 1$, we have
$k' \to 0$ and 
\begin{equation*}
\begin{array}{l}
K(k) = -\log{k'}+2\log{2} + \mathcal{O}((k')^2), \\
K(k') = \frac{\pi}{2} + \frac{\pi}{8}k'^2 + \mathcal{O}((k')^4), \\
q' = \frac{1}{16} k'^2 + \frac{1}{32} k'^4 + \mathcal{O}((k')^6).
\end{array}
\end{equation*}
These expansions simplify \eqref{eq:7} to
\begin{align*}
  \hat{f}(n) &= \sqrt{\frac{-2\log{k'}}{\pi}} \left (
               \frac{k'}{4} \right)^{\frac{(2n-1)^2}{2}} e^{(2n-1)x} \left ( 1 +
               \frac{x^2-2\log{2}}{\log{k'}} + \cdots \right ), \quad \mbox{\rm as} \;\; k' \to 0, 
\end{align*}
for every fixed $x \in \R$.  Then, the rightmost summation in \eqref{eq:9} yields the asymptotic expansion $\Theta(x) = \hat{f}(0) + \hat{f}(1) + \cdots$ 
in the form (\ref{theta4-asymptotics}). Using it in (\ref{tau-2soli-bright}), we obtain the asymptotic expansion 
\begin{equation}
\tau(x,t) \sim  \sqrt{\frac{-2 k' \log{k'}}{\pi}} \Big[ 
\cosh(\xi_1 - \delta_b)
           e^{\sqrt{|\lambda|}\xi_2} + \cosh(\xi_1 + \delta_b)
           e^{-\sqrt{|\lambda|} \xi_2} \Big],
           \label{tau-new}
\end{equation}
where $\xi_1 = x-4t$ and $\xi_2 = x - 4|\lambda| t + x_0$ for every $\lambda \in (-\infty,-1)$.  

Using \eqref{tau-new} with \eqref{delta-2soli-bright} in
\eqref{new-solution}, we obtain the two-soliton solution in the form:
\begin{align*}
  u(x,t) = 2 \frac{e^{2\delta_b}(1-\sqrt{|\lambda|})^2 +
    e^{-2\delta_b}(1+\sqrt{|\lambda|})^2 + 2\cosh ( 2\sqrt{|\lambda|} \xi_2
    ) + 2 |\lambda| \cosh ( 2\xi_1)}{[
    e^{\sqrt{|\lambda|}\xi_2} 
    \cosh(\xi_1-\delta_b) + e^{-\sqrt{|\lambda|}\xi_2}
    \cosh(\xi_1+\delta_b)]^2}.
\end{align*}
The two-soliton solution exhibits the asymptotic behavior
\begin{align*}
  u(x,t) \sim \ 2 \; \mathrm{sech}^2\left ( \xi_1 \mp \delta_b
                  \right ) + 2 |\lambda| \; \mathrm{sech}^2\left
                  (\sqrt{|\lambda|} \xi_2 \pm
                  \delta_b \right ), \quad \mbox{\rm as} \;\; t \to \pm \infty.
\end{align*}
After the interaction, the slower soliton of amplitude 2 experiences the
negative phase shift $-2 \delta_b,$ whereas the faster soliton of
amplitude $2 |\lambda|$ exhibits the positive phase shift
$2 \delta_b/\sqrt{|\lambda|}$.

\section{Properties of the dark breather}
\label{sec-7}

Figure \ref{dep dark} plots $\Delta_d$, $\kappa_d$, and $c_d$ for dark
breathers as a function of the parameter $\lambda$, see Theorem
\ref{th-dark} and Remark \ref{rem-dark}.  The phase shift $\Delta_d$
is monotonically decreasing between the band edges $\lambda_2(k)$ and
$\lambda_3(k),$ shown by the vertical dashed lines. The inverse width
$\kappa_d$ has a single maximum and vanishes at the band edges. The
breather speed $c_d$ is monotonically decreasing. Since $c_0 = -0.08$
for $k = 0.7$, we confirm that $c_d < c_0$, which is also clear from
Figure \ref{dark sol}.

\begin{figure}[htb!]
  \centering
  \includegraphics{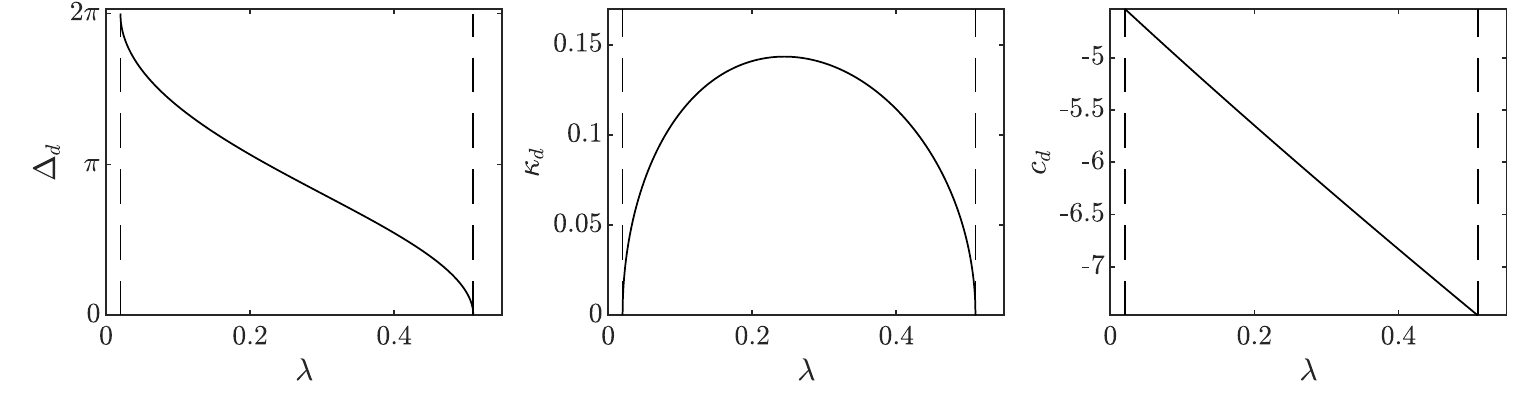}
  \caption{Normalized phase shift $\Delta_d$ (left), inverse width
    $\kappa_d$ (middle), and soliton speed $c_d$ (right) versus
    $\lambda$ in $(\lambda_2(k),\lambda_3(k))$ for $k = 0.7$. The band
    edges $\lambda_2(k) = 1-2k^2$ and $\lambda_3(k)=1-k^2$ are shown
    by the vertical dashed lines.}
  \label{dep dark}
\end{figure}

Figure \ref{fig:dark-breather-dispersion-modes} characterizes the
family of dark breathers by plotting $c_d - c_0$ and $\kappa_d$ versus
$\Delta_d$ for three values of $k$. The profiles of breather solutions
at $t = 0$ subject to the phase shift $x_0 = 5$ confirm why we refer
to them as dark breathers.  In contrast to the bright breather case,
dark breather solutions exhibit vanishing cnoidal wave modulations for
both of the extreme phase shifts $\Delta_d \to 0$ and
$\Delta_d \to 2\pi$, with the largest-amplitude breather occurring at
an intermediate phase shift, which we will later identify by examining
the inverse width $\kappa_d$.

\begin{figure}[htb!]
	\centering
	\includegraphics[height=9.7cm]{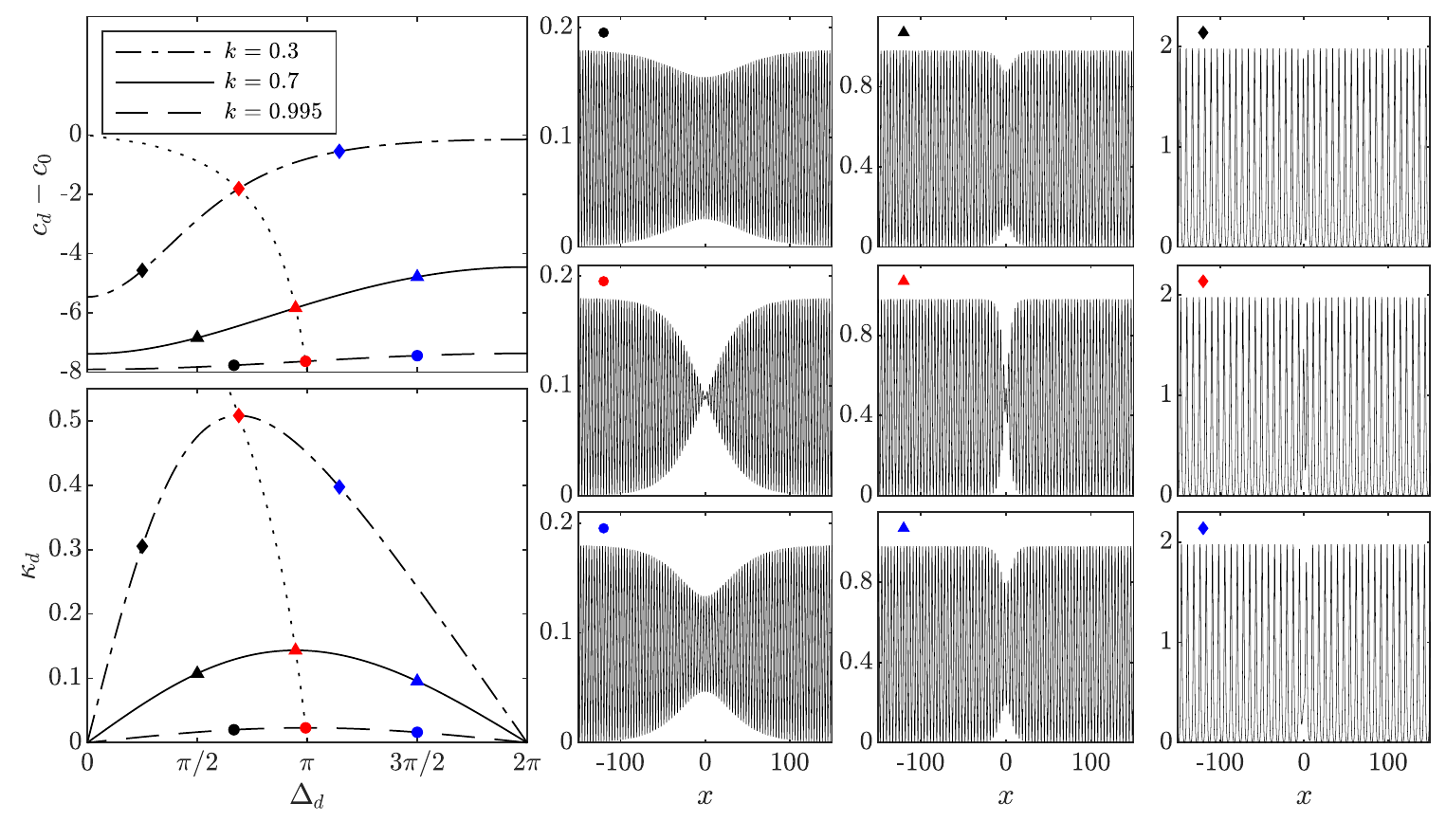}
	\caption{Left top (bottom): dependence of $c_d - c_0$ ($\kappa_d$)
      versus $\Delta_d$ for several values of $k$. Right:
      representative dark breather solutions. Representative solutions
      are marked on the left panel with a unique colored symbol. The
      dotted curve on the left panel corresponds to points of maximum
      $\kappa_d$ with the greatest localization.}
	\label{fig:dark-breather-dispersion-modes}
\end{figure}

\subsection{Asymptotic limits $\lambda \to \lambda_2(k)$ and
$\lambda \to \lambda_3(k)$} \hfill\\
\\
It follows from (\ref{phi-alpha}) that
\begin{equation*}
\varphi_{\alpha} = \left\{ 
  \begin{array}{ll} \displaystyle
    \frac{\pi}{2} - \frac{\sqrt{\lambda-\lambda_2(k)}}{k} +
    \mathcal{O}\left ( \lambda - \lambda_2 \right )
    \quad &
            \mbox{\rm as} \;\; \lambda \to \lambda_2(k), \\[3mm]  \displaystyle  
    \frac{\sqrt{\lambda_3(k) - \lambda}}{k} + \mathcal{O}\left (
    \lambda_3 - \lambda \right )
    \quad & \mbox{\rm as} \;\; \lambda \to \lambda_3(k) .
  \end{array}
\right.
\end{equation*}
The itemized list below summarizes the asymptotic results, where we
use the asymptotic equivalence for the leading-order terms and neglect
writing the remainder terms.

\begin{itemize}
\item The asymptotic values of the normalized phase shift $\Delta_d$ are 
  \begin{align*}
    \Delta_d = \left\{ 
    \begin{array}{ll} \displaystyle
      2\pi - \frac{2\pi}{K(k)}
      \sqrt{\frac{\lambda-\lambda_2(k)}{k^2(1-k^2)}}   \quad & \mbox{\rm as}
                                                         \;\; \lambda
                                                         \to
                                                         \lambda_2(k),
      \\  \displaystyle  
      \frac{2\pi}{K(k)} 
      \sqrt{\frac{\lambda_3(k)-\lambda}{k^2}} \quad & \mbox{\rm as} \;\;
                                               \lambda \to
                                               \lambda_3(k). 
    \end{array}
                                               \right.
  \end{align*}
  Since
  \begin{equation*}
      \partial_\lambda \Delta_d = \frac{2\pi}{K(k)}
                                  \partial_{\varphi_\alpha}
                                  F(\varphi_\alpha,k) \partial_\lambda 
                                  \varphi_\alpha 
  \end{equation*}
  with $\partial_{\varphi} F(\varphi,k) > 0$ and
  $\partial_{\lambda} \varphi_{\alpha} < 0$, the phase shift
  $\Delta_d$ monotonically decrease from $2\pi$ at
  $\lambda = \lambda_2(k)$ to $0$ at $\lambda = \lambda_3(k)$. This
  proves that the map $\lambda \mapsto \Delta_d(\lambda)$ is
  one-to-one and onto from
  $[\lambda_2(k),\lambda_3(k)]$ to $[0,2\pi]$.  \\

\item   The asymptotic values of the inverse width $\kappa_d$ are
\begin{align*}
\kappa_d = \left\{ 
\begin{array}{ll} \displaystyle
\left ( \frac{E(k)}{K(k)} -1+k^2 \right )
\sqrt{\frac{\lambda-\lambda_2(k)}{k^2(1-k^2)}}   \quad & \mbox{\rm as} \;\; \lambda \to \lambda_2(k), \\  \displaystyle
\left (
\frac{K(k)}{E(k)} - 1 \right )\frac{\sqrt{\lambda_3(k)-\lambda}}{k} \quad & \mbox{\rm as} \;\; \lambda \to \lambda_3(k).
\end{array}
\right.
\end{align*}
The inverse width $\kappa_d = Z(\varphi_\alpha,k)$
exhibits a maximum when \cite[eq.~141.25]{Byrd}
\begin{equation*}
  \sin \varphi_\alpha = \frac{1}{k}\sqrt{1-\frac{E(k)}{K(k)}} \quad \iff
  \quad \lambda = \lambda_{\rm max}(k) := \frac{E(k)}{K(k)} - k^2 .
\end{equation*}
The dark breather with this value of $\lambda$ can be interpreted as
the narrowest (strongest) modulation of the cnoidal wave.  Plotting
the behavior of $\Delta_{\rm max}(k) := \Delta_d$ at
$\lambda = \lambda_{\rm max}(k)$ as a function of $k$, we find that
$$0 < \Delta_{\rm max}(k) < \pi,$$ with the upper limit reached as
$k \to 0$.  The dotted curve in the left top panel of
Fig.~\ref{fig:dark-breather-dispersion-modes} shows the graph of
$$\left \{ (\Delta_{\rm max}(k),c_d(\lambda_{\rm max}(k))-c_0)\;\big
  | \; k \in (0,1) \right \}.$$ Consequently, the most localized dark
breather exhibits a positive phase shift. The phase shift is negative
for $\lambda$ near $\lambda_2(k)$ since $\Delta_d - 2\pi \in (-\pi,0)$
(cf.~Remark \ref{rem-dark}) and is positive for $\lambda$ near
$\lambda_3(k)$ since $\Delta_d \in (0,\pi)$.  This partitions dark
breathers into two branches: the slow (fast) branch for
$0 < \Delta_d < \Delta_{\rm max}(k) < \pi$
($\Delta_{\rm max}(k) < \Delta_d < 2\pi$).  The slow branch exhibits
dark breathers with strictly positive phase shifts whose amplitudes
increase with increasing phase shift.  On the fast branch, dark
breathers can have positive or negative phase shift depending on
whether $\Delta_d$ is less than or greater than $\pi$, respectively.
Also, an increase in phase shift corresponds to a decrease in
amplitude.  \\

\item The asymptotic values of the breather speed $c_d$ are
\begin{equation*}
c_d = \left\{ 
\begin{array}{ll} \displaystyle
c_0 - \frac{4k^2(1-k^2)}{E(k)/K(k) - 1 + k^2}
\quad & \mbox{\rm as} \;\; \lambda \to \lambda_2(k), \\
\displaystyle c_0 - \frac{4k^2}{1-E(k)/K(k)}
\quad & \mbox{\rm as} \;\; \lambda \to \lambda_3(k).
\end{array}
\right.
\end{equation*}
Based on the graphs in Fig.~\ref{dep dark}, we conjecture that the breather
velocity $c_d$ is a monotonically decreasing function of $\lambda$.
\end{itemize}

\subsection{Asymptotic limit $k \to 1$}
\label{sec:solit-limit-backgr}\hfill\\
\\
We show similarly to Section \ref{sec:harm-solit-limits} that 
the dark breather recovers the two-soliton
solution in the limit $k \to 1$. The only difference 
from the degeration of the bright breather is that 
the spectral parameter $\lambda$ is now defined in $(-1,0)$ 
rather than in $(-\infty,-1)$. By using \eqref{theta4-asymptotics} 
in (\ref{tau-function-dark}), we obtain the asymptotic approximation 
\begin{equation}
  \label{tau-2soli-dark}
  \tau(x,t) \sim \sqrt{\frac{-2k'\log{k'}}{\pi}} \Big[
  \cosh(\xi_1+\delta_d) e^{-\sqrt{|\lambda|}\xi_2} +
  \cosh(\xi_1-\delta_d) e^{\sqrt{|\lambda|}\xi_2} \Big], \quad \mbox{\rm as} \;\; k \to 1,
\end{equation}
where $\xi_1 = x -4t$ and $\xi_2 = x - 4 |\lambda| t + x_0$ for $\lambda \in (-1,0)$ and we have used $\kappa_d \to \sqrt{|\lambda|}$, 
$c_d \to 4 |\lambda|$, and the corresponding limiting phase $\delta_d$ found from 
\cite[eq.~(2.7)]{VandeVel}:
\begin{equation}
  \label{delta-2soli-dark}
  \delta_d := \lim_{k\to 1} F(\varphi_\alpha,k) = \frac{1}{2} \log
  \left ( \frac{1+\sqrt{|\lambda|}}{1-\sqrt{|\lambda|}} \right ),
  \quad \lambda \in (-1,0) .
\end{equation}
Inserting \eqref{tau-2soli-dark} and \eqref{delta-2soli-dark} into
\eqref{new-solution} results in the two-soliton solution
\begin{align*}
  u(x,t) =2 \frac{e^{2\delta_d}(1-\sqrt{|\lambda|})^2 +
    e^{-2\delta_d}(1+\sqrt{|\lambda|})^2 + 2\cosh ( 2\sqrt{|\lambda|} \xi_2
    ) + 2|\lambda| \cosh ( 2\xi_1)}{[e^{-\sqrt{|\lambda|}\xi_2} 
    \cosh(\xi_1+\delta_d) + e^{\sqrt{|\lambda|}\xi_2}
    \cosh(\xi_1-\delta_d) ]^2} ,
\end{align*}
that exhibits the asymptotic behavior
\begin{align*}
  u(x,t) \sim \ &2 \; \mathrm{sech}^2\left ( \xi_1 \pm \delta_d
                  \right ) + 2 |\lambda| \; \mathrm{sech}^2\left
                  (\sqrt{|\lambda|} \xi_2 \mp
                  \delta_d \right ), \quad t \to \pm \infty .
\end{align*}
After the interaction, the slower soliton of amplitude $2|\lambda|$
experiences the negative phase shift $-2 \delta_d/\sqrt{|\lambda|}$
whereas the faster soliton of amplitude $2$ exhibits the positive
phase shift $2 \delta_d$.
  
\subsection{Asymptotic limit $k \to 0$}
\label{sec:small-ampl-solit}\hfill\\
\\
We show that the dark breather as $k \to 0$ can be approximated by a dark
soliton solution of the nonlinear Schrodinger (NLS) equation.

In the limit $k \to 0$, the interval $[\lambda_2(k),\lambda_3(k)]$ shrinks to
the point $\lambda = 1$ and the solution $u(x,t)$  converges to the zero solution such that both the cnoidal wave and the dark breather vanish.  For small $k$,
it is well-known (see, e.g., \cite{Zakharov}) that the multiple scales expansion 
\begin{align}
u(x,t) &= 2 \mathrm{Re} \bigg[ \epsilon \sqrt{\frac{\ell}{6}}
             A(\zeta,\tau) e^{i(\ell x - \omega t)} \notag \\
       & \qquad\quad\,     + \epsilon^2 \frac{\ell}{6} \Big (
             \frac{1}{4} A(\zeta,\tau)^2 
             e^{2i(\ell x - \omega t)} - \frac{1}{2} |A(\zeta,\tau)|^2 \Big) +
             \mathcal{O}(\epsilon^3)  \bigg]
  \label{nls-expansion}
\end{align}
leads to the following NLS equation for the slowly varying amplitude
$A(\zeta,\tau)$:
\begin{equation}
  \label{nls}
  i A_\tau - \frac{1}{2} A_{\zeta\zeta} + |A|^2 A = 0 ,
\end{equation}
where $0 < \epsilon \ll 1$ is the amplitude parameter, $\ell > 0$ is the
carrier wavenumber, $\omega = - \ell^3$ is the KdV linear dispersion relation,
and $\zeta = \frac{\epsilon}{\sqrt{6 \ell}} (x + 3 \ell^2 t)$ and
$\tau = \epsilon^2 t$ are slow variables.  The NLS equation
\eqref{nls} admits the plane wave solution
\begin{equation}
  \label{nls-plane-wave}
  A(\zeta,\tau) = e^{i (1 + \frac{v^2}{2})\tau + i v \zeta + i \psi_0}
\end{equation}
for any $v,\psi_0 \in \R$.  To determine $\ell$ and $\epsilon$, it is
necessary to expand the cnoidal wave background of the dark breather
solution for small elliptic modulus $0 < k \ll 1$:
\begin{align*}
    u(x,t) &= 2k^2 {\rm cn}^2(x-c_0t) \\
&=  2k^2 \cos^2(x - c_0 t) + \mathcal{O}(k^4) \\
&= k^2 + k^2 \cos 2 (x - c_0 t) + \mathcal{O}(k^4) , 
\end{align*}
where $c_0 \to -4$ as $k \to 0$. The background cnoidal wave's
wavenumber $Q$, frequency $\Omega$, and mean value $\overline{\phi}$
expand as $k \to 0$ in the form:
\begin{equation*}
\begin{array}{l}
    Q := \frac{\pi}{K(k)} = 2 - \frac{k^2}{2} + \mathcal{O}(k^4), \\
        \Omega  := c_0 Q = -8 + 18 k^2 + \mathcal{O}(k^4), \\
    \overline{\phi} := \frac{1}{2K(k)} \int_0^{2K(k)} \phi_0(x)\,dx =
                      k^2 + \mathcal{O}(k^4).
                      \end{array}
\end{equation*}
Comparing \eqref{nls-expansion} with the asymptotic expansion 
for the background cnoidal wave, we find $\epsilon = \frac{k^2\sqrt{3}}{2}$ 
and $\ell = 2$ confirming that the limit $k \to 0$ coincides with the NLS
approximation.  Since the expansion \eqref{nls-expansion} does not
incorporate an $\mathcal{O}(\epsilon)$ mean term, the Galilean transformation
of the KdV equation can be used in \eqref{nls-expansion} and \eqref{nls-plane-wave} to obtain 
\begin{equation}
  \label{plane-wave-expansion}
    u(x,t) \to k^2 + u(x-6k^2t,t) = k^2 + k^2 \cos ( \Lambda x - \Upsilon t + \psi_0 ) + \mathcal{O}(k^4) , 
 \end{equation}
 where
\begin{equation*}
\begin{array}{l}
    \Lambda = 2+ \frac{v}{4} k^2 + \mathcal{O}(k^4), \\ 
    \Upsilon = -8 + (12-3v)k^2
              + \mathcal{O}(k^4).
              \end{array}
\end{equation*}
The choice $v = -2$ asymptotically matches $\Lambda$ and $\Upsilon$ in
\eqref{plane-wave-expansion} with $Q$ and $\Omega$.

The NLS equation \eqref{nls} admits two families of dark soliton
solutions \cite{Ablowitz}
\begin{equation}
  \label{nls-dark-soliton}
  A(\zeta,\tau) = \Big ( \cos{\beta} \pm i \sin \beta \tanh \big (
  \sin \beta (\zeta - c_\pm \tau) \big ) \Big ) e^{i\left (
      - 2 \zeta + 3\tau + \psi_0 \right )},
\end{equation}
where $\pm$ corresponds to the fast ($+$) and slow ($-$) solution
branches with velocities $c_\pm = 2 \pm \cos{\beta}$, phase shift
parameter $\beta \in [0,\pi/2]$, and arbitrary phase $\psi_0 \in \R$.
Since 
$$
A(\zeta,\tau) \to e^{i(-2\zeta + 3 \tau + \psi_0 \mp \beta)} \quad 
\mbox{\rm as} \;\; \zeta - c_\pm \tau \to -\infty
$$
and
$$
A(\zeta,\tau) \to e^{i(-2\zeta + 3 \tau + \psi_0 \pm \beta)} \quad 
\mbox{\rm as} \;\; \zeta - c_\pm \tau \to \infty, 
$$
the normalized phase shift is
$\Delta_{\pm} := \pm 2 \beta$ for the fast ($-$) and slow ($+$) branch
of solutions.  Applying the Galilean transformation
$u(x,t) \to k^2 + u(x-6k^2t,t)$ to Eqs.~\eqref{nls-expansion} and
\eqref{nls-dark-soliton}, the dark soliton velocity-phase shift
relation $c_\pm$ is
\begin{equation}
  \label{dark-soliton-speed}
  \begin{split}
    c_\pm = -12 + \bigg ( 12 + 3 \, \mathrm{sgn}(\Delta_\pm) \cos \Big (
    \frac{\Delta_\pm}{2} \Big ) \bigg )k^2, \quad \Delta_+ \in
    (0,\pi], \quad \Delta_- \in (-\pi,0) .
  \end{split}
\end{equation} 
From Eq.~\eqref{nls-dark-soliton}, the inverse width parameter
$\kappa_\pm := \sin \frac{\beta \epsilon}{\sqrt{12}}$ is given by 
\begin{equation}
  \label{dark-soliton-width}
  \kappa_\pm = \frac{1}{4} \sin \left ( \frac{|\Delta_\pm|}{2} \right )k^2 .
\end{equation}

In order to compare the dispersion relation given by \eqref{dark-soliton-speed} and \eqref{dark-soliton-width} with the dark breather dispersion relation 
given by 
\eqref{eq:25}, \eqref{eq:24}, \eqref{dark-breather-phase-shift}, we expand
the spectral parameter $\lambda$ as $\lambda = 1-k^2(1+\mu)$ with new scaled 
spectral parameter $\mu$,
ensuring a distinct breather for each $\mu \in (0,1)$ as $k \to 0$.
The small $k$ expansion of the dark breather dispersion relation
\eqref{eq:6}, \eqref{eq:25}, \eqref{eq:24}, and
\eqref{dark-breather-phase-shift} is given by 
\begin{equation*}
\begin{array}{l}
    \alpha_d = \arcsin(\sqrt{\mu}) + \mathcal{O}(k^2), \\
               \kappa_d = \frac{1}{2} \sqrt{\mu(1-\mu)} k^2 +
               \mathcal{O}(k^4), \\
    c_d = -12 + (9 + 6 \mu) k^2 + \mathcal{O}(k^4), \\
          \Delta_d = 4 \arcsin(\sqrt{\mu}) + \mathcal{O}(k^2),
          \end{array}
  \end{equation*}
for $\mu \in [0,1]$.  Substituting $\mu = \sin^2 \left (
  \frac{\Delta_d}{4} \right )$ yields
\begin{equation}
  \label{dark-breather-dispersion-small-k}
  \begin{array}{l}
  c_d = -12 + \left[ 12 - 3 \cos \Big (\frac{\Delta_d}{2} \Big )
  \right] k^2 + \mathcal{O}(k^4), \\ \kappa_d = \frac{1}{4}
  \sin \left ( \frac{|\Delta_d|}{2} \right )k^2 + \mathcal{O}(k^4).
  \end{array}
\end{equation}
By identifying certain values of the phase shift $\Delta_d$ with the
slow ($+$) and fast ($-$) branches of the NLS dark soliton solution
\eqref{nls-dark-soliton}, as given by 
\begin{align*}
  \Delta_- = \Delta_d, \quad & \Delta_d \in (0,\pi], \\
  \Delta_+ = \Delta_d - 2\pi, \quad & \Delta_d \in (-\pi,0),
\end{align*}
we find that Eq.~\eqref{dark-breather-dispersion-small-k} coincides
with Eqs.~\eqref{dark-soliton-speed} and \eqref{dark-soliton-width} up
to and including the $\mathcal{O}(k^2)$ terms.  The fast and slow
branches of the NLS dark soliton \eqref{nls-dark-soliton} coincide
with the limiting fast and slow branches of the dark breather.  The
black soliton solution \eqref{nls-dark-soliton} with $\beta = \pi/2$
corresponds to the dark breather of maximum localization in which
$\Delta_{\rm max}(k) \sim \pi$ as $k \to 0$.

\section{Conclusion}
\label{sec-8}

A comprehensive characterization of explicit solutions of the KdV
equation, representing the nonlinear superposition of a soliton and
cnoidal wave, has been obtained using the Darboux transformation.  These
solutions are breathers, manifesting as nonlinear wavepackets
propagating with constant velocity on a cnoidal, periodic, traveling
wave background, subject to a topological phase shift.  Breathers of
elevation type, called bright breathers, are shown to propagate faster
than the cnoidal background.  Depression-type breathers are called
dark breathers and they move slower than the cnoidal background. A key finding is that each breather on a fixed cnoidal wave background is
uniquely determined by two distinct parameters: its initial position
and a spectral parameter. We prove that the spectral parameter is in
one-to-one correspondence with the normalized phase shift, which it imparts to the cnoidal background, in the
interval $(-\pi,\pi].$

Bright breathers with small, negative phase shifts correspond to
small-scale amplitude modulations of the cnoidal wave background,
which result in the cnoidal wave dominating the solution. Small,
positive phase shifts correspond to bright breathers with large-scale
amplitude modulations of the cnoidal wave background where the soliton
component is dominant.  As the phase shift is swept across the
interval $(-\pi,\pi]$, all breather amplitudes are attained.

In contrast, dark breather amplitudes, being of depression type, are
limited.  Small phase shifts, positive or negative, correspond to
small modulations of the cnoidal wave background and the slow or fast
branch of solutions, respectively.  For each cnoidal wave background,
we find a narrowest dark breather that imparts a positive phase
shift. When the amplitude of the cnoidal wave background is small,
dark breathers degenerate into dark soliton solutions of the NLS
equation (\ref{nls}) derived from the KdV equation (\ref{kdv}).

When the period of the cnoidal wave background goes to infinity, both
bright and dark breather solutions are shown to degenerate into
two-soliton solutions of the KdV equation.  In this sense, breathers
can be viewed as a generalization of two-soliton interactions.  While
such an interpretation is well-known for the sine-Gordon, focusing
NLS, and the focusing modified KdV equations where breathers can be
interpreted as bound states of two solitons \cite{AblowitzSegur},
those breather solutions are localized.  In contrast, the topological
KdV breathers with an extended, periodic background described here
represent a different class of nonlinear wave interaction solutions.
We expect that such solutions exist for other integrable nonlinear
evolutionary equations with a self-adjoint scattering problem such as
the defocusing NLS and defocusing modified KdV equations.

An important application of these breather solutions is to the problem
of soliton-dispersive shock wave (DSW) interaction \cite{Maiden}.
Bright breathers were identified in \cite{Cole} as being associated
with soliton-DSW transmission. Soliton-DSW trapping corresponds to
dark breathers embedding within the DSW.  The spectral
characterization of KdV breathers obtained here can be used in the
context of multi-phase Whitham modulation theory \cite{Flaschka} to
describe the dynamics of breathers subject to large-scale amplitude
modulations \cite{Cole}.  In addition to soliton-DSW interaction, the
bright breathers resemble the propagation of a soliton through a
special kind of deterministic soliton gas, constructed using
Riemann-Hilbert methods from primitive potentials of the defocusing
modified KdV equation \cite{Girotti}.  Similar deterministic soliton
gases have been identified as soliton condensates for the KdV equation
\cite{CongyEl} and provide further applications for the breathers
constructed here.

\vspace{0.25cm}

{\bf Acknowledgement.}  The authors would like to thank the Isaac
Newton Institute for Mathematical Sciences for support and hospitality
during the programme \textit{Dispersive Hydrodynamics} when work on
this paper was undertaken (EPSRC Grant Number EP/R014604/1). The
authors thank Y. Kodama and G. El for many useful suggestions on this
project.  MAH gratefully acknowledges support from NSF DMS-1816934.

\bgroup
\def\arraystretch{2}
\begin{table}[h]
	\centering
	\begin{tabular}{|c|l|}
		\hline
		$F(\varphi,k)$  & Elliptic integral of the first kind 
		$\displaystyle F(\varphi,k) := \int_0^{\varphi} \frac{d \alpha}{\sqrt{1-k^2 \sin^2 \alpha}}$ \\
		\hline
		$E(\varphi,k)$  & Elliptic integral of the second kind 
		$\displaystyle E(\varphi,k) := \int_0^{\varphi} \sqrt{1-k^2 \sin^2 \alpha} d\alpha$\\
		\hline
		$K(k)$ & Complete elliptic integral
		$\displaystyle K(k) := F\left(\frac{\pi}{2},k\right)$ \\
		\hline
		$E(k)$ & Complete elliptic integral
		$\displaystyle E(k) := E\left(\frac{\pi}{2},k\right)$\\
		\hline
		$Z(\varphi,k)$ & Zeta function 
		$\displaystyle Z(\varphi,k) := E(\varphi,k) - \frac{E(k)}{K(k)} F(\varphi,k)$\\
		\hline
		${\rm sn}(x,k)$  & Jacobi elliptic function  ${\rm sn}(x,k) = \sin \varphi $\\
		\hline
		${\rm cn}(x,k)$  & Jacobi elliptic function  ${\rm cn}(x,k) = \cos \varphi $\\
		\hline
		${\rm dn}(x,k)$  & Jacobi elliptic function ${\rm dn}(x,k) = \sqrt{1 - k^2 \sin^2 \varphi}$\\
		\hline
		$H(x)$  & $\theta_1\left(\frac{\pi x}{2 K(k)}\right)$ with 
		$\displaystyle \theta_1(u) = 2 \sum\limits_{n=1}^{\infty} (-1)^{n-1} q^{(n-\frac{1}{2})^2} \sin(2n-1) u$ \\
		\hline
		$H_1(x)$  & $\theta_2\left(\frac{\pi x}{2 K(k)}\right)$ with 
		$\displaystyle \theta_2(u) = 2 \sum\limits_{n=1}^{\infty} q^{(n-\frac{1}{2})^2} \cos(2n-1) u$  \\
		\hline
		$\Theta_1(x)$  & $\theta_3\left(\frac{\pi x}{2 K(k)}\right)$ with 
		$\displaystyle \theta_3(u) = 1 + 2 \sum\limits_{n=1}^{\infty}  q^{n^2} \cos 2nu$  \\
		\hline
		$\Theta(x)$  & $\theta_4\left(\frac{\pi x}{2 K(k)}\right)$ with 
		$\displaystyle \theta_4(u) = 1 + 2 \sum\limits_{n=1}^{\infty} (-1)^{n} q^{n^2} \cos 2n u$  \\
		\hline
		$q$ & $\displaystyle e^{-\frac{\pi K'(k)}{K(k)}}$ with $K'(k) = K(k')$ and $k' = \sqrt{1-k^2}$ \\
		\hline
	\end{tabular}
	\vspace{0.15cm} 
	\caption{Table of elliptic integrals and Jacobi elliptic functions.}
	\label{Table-1}
\end{table}
\egroup

\end{document}